\documentclass[a4paper, 12pt]{article}

\usepackage{authblk}
\usepackage{float}
\usepackage[a4paper,
			top=3cm, bottom=3cm,
		    left=3cm, right=3cm]{geometry}

\usepackage[english]{babel}
\usepackage{amsmath, amsfonts, amssymb, amsthm}
\usepackage[colorlinks = true,
			linkcolor = blue,
			urlcolor = black,
			bookmarksopen = true]{hyperref}				
\usepackage[capitalise]{cleveref}
\usepackage[
backend=biber,
style=numeric,
sorting=nyt,
maxalphanames=3,
maxbibnames=3,
isbn=true,
url=true,
doi=true,
giveninits=true,
]{biblatex}
\usepackage{csquotes}
\usepackage{longtable}
\usepackage{tikz}
\usetikzlibrary{positioning}
\usetikzlibrary{calc}
\theoremstyle{definition}
\newtheorem{theo}[equation]{Theorem}
\crefname{theo}{Theorem}{Theorems}
\newtheorem{defi}[equation]{Definition}
\newtheorem{prop}[equation]{Proposition}
\newtheorem{algo}[equation]{Algorithm}
\newtheorem{lem}[equation]{Lemma}
\newtheorem{cor}[equation]{Corollary}
\newtheorem{rem}[equation]{Remark}
\numberwithin{equation}{section}

\usepackage{graphicx}
\graphicspath{{./img/}}

\newcommand{\im}{\mathrm{i}}
\newcommand{\ket}[1]{\left| #1 \right\rangle}
\newcommand{\ex}[1]{\exp\left( #1 \right)}
\DeclareMathOperator{\gl}{GL}

\usepackage{pifont}

\title{A survey about Hidden Subgroup Problem from a mathematical and cryptographic perspective}
\author[1]{Simone Dutto}
\author[2]{Pietro Mercuri} 
\author[2]{Nadir Murru} 
\author[3]{Lorenzo Romano}

\affil[1]{Independent Researcher\\
\texttt{smn.dutto@gmail.com}}
\affil[2]{Department of Mathematics, University of Trento\\
\texttt{mercuri.ptr@gmail.com, nadir.murru@unitn.it}}
\affil[3]{Department of Mathematical Sciences, Politecnico of Torino\\
\texttt{lorenzo.romano@polito.it}}
\date{}

\begin{document}

\maketitle

\begin{abstract}
We provide a survey on the Hidden Subgroup Problem (HSP), which plays an important role in studying the security of public-key cryptosystems.
We first review the abelian case, where Kitaev's algorithm yields an efficient quantum solution to the HSP, recalling how classical problems (such as order finding, integer factorization, and discrete logarithm) can be formulated as abelian HSP instances.
We then examine the current state of the art for non-abelian HSP, where no general efficient quantum solution is known, focusing on some relevant groups including dihedral group (connected to the shortest vector problem), symmetric groups (connected to the graph isomorphism problem), and semidirect product constructions (connected, in a special case, to the code equivalence problem).
We also describe the main techniques for addressing the HSP in non-abelian cases, namely Fourier sampling and the black-box approach.
Throughout the paper, we highlight the mathematical notions required and exploited in this context, providing a cryptography-oriented perspective.
\end{abstract}

\section{Introduction}
\label{sec:intro}

Public-key cryptography relies on assumptions that guarantee the security of a scheme up to a certain level, related to the computational complexity of solving the underlying mathematical problem.
Such a resource-dependent definition is severely damaged by the awakening of quantum computers, on which new algorithms are available and old computational assumptions no longer hold.
In this context, the Hidden Subgroup Problem (HSP) appears to be one of the minimal cryptographic assumptions, in the sense that most of the widely used assumptions can be efficiently recast as an instance thereof (e.g., discrete logarithm, integer factorization, but also the shortest vector problem and others).
The post-quantum relevance of this assumption is motivated by Kitaev's algorithm, a quantum routine which generalizes Shor's algorithm and efficiently solves the HSP for abelian groups,
delimiting the area of post-quantum to those assumptions that do not reduce at all or for which the only known reductions are to non-abelian cases.
Indeed, currently, there are no efficient quantum algorithms for solving the HSP in the non-abelian case, despite the existence of several approaches focusing on specific non-abelian groups, such as the dihedral group (which is connected to the shortest vector problem).
In this paper, we would like to provide an overview about the HSP, showing which problems of cryptographic relevance reduce to instances of this problem and surveying the existing quantum algorithms for solving the HSP across the groups studied so far (distinguishing between abelian and non-abelian cases).
Although HSP plays a central role in the study of cryptosystem security, existing surveys are often written from a computer science perspective, assuming familiarity with quantum computing and omitting mathematical details.
This work aims at addressing this gap highlighting the mathematical notions required and exploited in this context, and providing a cryptography oriented presentation.
Thus, the survey is intended for readers with a mathematical and cryptographic background who may lack basic knowledge about quantum computing and require an algebraic perspective on the problem and its implications.

\cref{sec:bg} provides the mathematical background and fixes the notation exploited in the paper.
Specifically, in Sections \ref{ssec:group} and \ref{ssec:reps}, we recall the basic notions needed in group theory and representation theory.
\cref{ssec:poly} discusses the concept of reductions and computational time of algorithms, while \cref{ssec:quantum} focuses on quantum computing, introducing the fundamentals necessary to understand the approaches used in the development of quantum algorithms.
Finally, in \cref{ssec:hsp} we formally introduce the HSP.

In \cref{sec:abel}, first, we outline Kitaev's algorithm, which was the first polynomial--time quantum algorithm able to solve the HSP for every abelian group, where the quantum routine on which it relies is called Fourier sampling and is reviewed in Algorithm \ref{algo:AHSP}.
In \cref{ssec:simon,ssec:ofp,ssec:dlp}, we then show how main classical problems, like Simon's problem, order finding problem, integer factorization, and the discrete logarithm problem, can be reduced to instances of abelian HSP.
\cref{ssec:root} is devoted to the discrete root problem, which currently has not yet been explicitly reduced to an HSP instance.
Since extracting a discrete root can be efficiently solved when the modulus factorization is known, this problem is already efficiently solvable by quantum algorithms and, consequently, cryptosystems relying on it are not quantum-safe.
However, a direct reduction of this problem to a HSP instance is interesting both from a computational point, since it could lead to improvements in the efficiency of the algorithms, and from a theoretical point of view, since it could show other ways to find reductions.
In this section, we present some of our attempts in this direction and we highlight the obstructions that prevented our approaches from succeeding.

\cref{sec:nabel} is devoted to the non--abelian case for which a quantum algorithm able to efficiently solve the general HSP is not known.
However, efficient methods exist for certain specific non-abelian groups, such as Dedekind and Hamiltonian groups; see \cref{ssec:ded}.
In \cref{ssec:dihedral}, we focus on the dihedral group which plays a significant role in cryptography, as the shortest vector problem on lattices can be reduced to the HSP over these groups.
Another relevant case is provided by the symmetric group, which is related to the graph isomorphism problem, see \cref{ssec:sym}.
In \cref{ssec:semi}, we present several results concerning groups obtained via the semidirect product.
\cref{ssec:code} focuses on a special semidirect product connected to the code equivalence problem, which has a well-known relevance in code--based cryptography.
The algorithms developed to address the HSP in the cases discussed in these sections are all based on the so-called Fourier sampling.
In \cref{ssec:other} we review other groups for which the HSP has been studied with this technique and \cref{ssec:general} reviews the state of the art for the general case. Finally, in \cref{ssec:black} we present an alternative approach to the Fourier sampling, namely the black-box approach.

\cref{sec:concl} contains the conclusions and also outlines some further developments regarding the HSP.

Finally, the Appendix provides a summary of the groups for which the HSP has been investigated, as well as a recap about the problems that can be seen as instances of the HSP.

\section{Background and notation}
\label{sec:bg}
In this section, we fix the notation used throughout the paper and provide the necessary background in group theory, representation theory, and quantum computing.

\subsection{Group theory} 
\label{ssec:group}

We start by recalling some basic concepts and definitions of group theory in order to fix notation and terminology.
For a more detailed source, see \cite{Rot12}.

Given the (abelian) group $(\mathbb Z,+)$ of integers, the subset $N\mathbb Z$ of all the multiples of an integer $N$ is a (normal) subgroup of $\mathbb Z$.
We denote by $\mathbb Z_N$ the quotient group $\mathbb Z/N\mathbb Z$.
Each element of $\mathbb Z_N$ represents the set of integers with a given remainder with respect to the division by $N$.
We usually denote the $N$ elements of $\mathbb Z_N$ by $0,1,\ldots, N-1$ that correspond to the cosets $0+N\mathbb Z,1+N\mathbb Z,\ldots,N-1+N\mathbb Z$ respectively.
\begin{rem}
An integer $m\in\mathbb Z$ can be sent naturally in $\mathbb Z_N$ just taking its coset:
\begin{align*}
\pi_N\colon\mathbb Z&\longrightarrow\mathbb Z_N, \\
m&\longmapsto m+N\mathbb Z.
\end{align*}
But the map in the other direction is not uniquely defined: each element of the set $m+N\mathbb Z$ can be taken as the image of the map $\mathbb Z_N\longrightarrow\mathbb Z$.
We fix these maps once and for all:
\begin{align*}
\mathrm{lift}_N\colon\mathbb Z_N&\longrightarrow\mathbb Z,\\
m+N\mathbb Z&\longmapsto k\quad\text{with }k\in\{0,1,\ldots,N-1\} \text{ such that }k-m\in N\mathbb Z.
\end{align*}
In other words, we always choose the integer $k$ among $0,1,\ldots,N-1$ such that $k$ and $m$ have the same remainder with respect to the division by $N$.

Hence, if we have an element in $\mathbb Z_N$ that we consider in $\mathbb Z_M$, for any two integers $N$ and $M$, we mean that we are using the following map
\[
\begin{tikzpicture}
  \node (A) at (0,0) {$\mathbb{Z}_N$};
  \node (B) at (2,0) {$\mathbb{Z}$};
  \node (C) at (4,0) {$\mathbb{Z}_M,$};
  \node (D) at (-0.5,-1) {$x + N\mathbb{Z}$};
  \node (E) at (5.5,-1) {$\mathrm{lift}_N(x + N\mathbb{Z}) + M\mathbb{Z}.$};

  \draw[->] (A) -- (B) node[midway, above] {$\mathrm{lift}_N$};
  \draw[->] (B) -- (C) node[midway, above] {$\pi_M$};
  \draw[|->] (D) -- (E) node[midway, above] {$\pi_M\circ \mathrm{lift}_N$};
\end{tikzpicture}
\]
\end{rem}
\begin{defi}
A group is \emph{finite} if it has finitely many elements.
The \emph{order} $|G|$ of a finite group $G$ is its number of elements.
The \emph{order} of an element $g$ in a group $G$ is the smallest positive integer $n$ such that $g^n=1_G$; if this never happens, the order of $g$ is infinity (this cannot happen in finite groups).\\
A \emph{set of generators} of a group $G$ is a set $B$ of elements of $G$ such that every element of $G$ can be written as a product of powers of finitely many elements of $B$.\\
A group $G$ is \emph{cyclic} if there is an element $g\in G$ such that $\{g^n,n\in\mathbb Z\}=G$, i.e., a group is cyclic if it can be generated by one (not unique) element.
\end{defi}

\begin{rem}
The group $\mathbb Z_N$ is cyclic of order $N$ and every integer coprime with $N$ is a generator.
\end{rem}

There are many different types of ``products" that given two groups produce a new one which can be defined.
In the following, we describe the main ones.
We start with the most basic product: the direct product.
\begin{defi}
The \emph{direct product} $G_1\times G_2$ of the two groups $G_1$ and $G_2$ is the group whose elements are the ordered pairs $(g_1,g_2)$, with $g_1\in G_1$ and $g_2\in G_2$, and the group operation is given componentwise, i.e.,
\[
(g_1,g_2)\cdot(h_1,h_2):=(g_1h_1,g_2h_2),\quad\text{for }g_1,h_1\in G_1, g_2,h_2\in G_2.
\]
This notion can be generalized for every finite number of groups $G_1\times G_2\times\cdots\times G_m$.
\end{defi}
A slightly more complicated product is the semidirect product; it is a generalization of the previous construction.
\begin{defi}
Let $N$ and $H$ be two groups, and given a group homomorphism $\varphi\colon H\longrightarrow\mathrm{Aut}(N)$, the \emph{semidirect product} $N\rtimes_\varphi H$ of $N$ and $H$ with respect to $\varphi$ is the group whose elements are the ordered pairs $(n,h)$, with $n\in N$ and $h\in H$, and the group operation is given by
\[
(n_1,h_1)\cdot(n_2,h_2):=(n_1\varphi(h_1)(n_2),h_1h_2),\quad\text{for }n_1,n_2\in N, h_1,h_2\in H.
\]
When $\varphi$ is the constant homomorphism $\varphi(h)=\mathrm{id}_N$, for each $h\in H$, we get the direct product.
When $H$ and $N$ are subgroups of $G = N H$ with $N$ normal and $H \cap N = \{1\}$, then there is a standard homomorphism on $N$ given by the conjugation action of $H$ that allows to obtain $G$ as the \textit{inner semidirect product} $G = N \rtimes H$.
\end{defi}

\begin{rem}
The dihedral group $D_N$ of the symmetries of a regular polygon with $N$ sides is called \emph{dihedral group} and it is finite of order $2N$.
It can be written, using the semidirect product, as
\[
D_N\cong \mathbb Z_N\rtimes_\varphi \mathbb Z_2,
\]
with $\varphi\colon\mathbb Z_2\longrightarrow\mathrm{Aut}(\mathbb Z_N)$ given by $\varphi(x)\colon n\longmapsto (-1)^x n$, for $n\in \mathbb Z_N$.
\end{rem}

The structure of the finite abelian groups is completely known, as it is shown by the following theorem.
\begin{theo}[Fundamental Theorem of Finite Abelian Groups]\label{th:FTFAG}
Every finite abelian group $G$ can be written as
\[
G\cong \mathbb Z_{N_1}\times\mathbb Z_{N_2}\times\cdots\times\mathbb Z_{N_\ell},
\]
where $N_1,N_2,\ldots,N_\ell$ are powers of (not necessarily distinct) primes.
\end{theo}

\subsection{Representation theory} 
\label{ssec:reps}

In the following, useful results from representation theory are recalled.
For a more complete and detailed reference on representations of finite groups, see, for example, \cite{Ser96}.

\begin{defi}
A \emph{representation} $(\rho,V)$ of a group $(G, \cdot)$ on a vector space $V$ over a field $K$ is a homomorphism $\rho\colon G \longrightarrow \mathrm{GL}(V)$, such that $\rho(g_1 \cdot g_2) = \rho(g_1) \rho(g_2),$ for each $g_1 , g_2 \in G$.
A \emph{complex representation} is a representation such that $K=\mathbb C$.\\
A \emph{homomorphism} between $\rho_1\colon G \longrightarrow \mathrm{GL}(V_1)$ and $\rho_2\colon G \longrightarrow \mathrm{GL}(V_2)$ is a linear map $\varphi\colon V_1 \longrightarrow V_2$ such that $\big(\varphi\circ \rho_1(g)\big)({v}) = \big(\rho_2(g)\circ\varphi\big)({v})$ for each $g \in G$ and $v\in V_1$.\\
Given a representation $\rho\colon G \longrightarrow \mathrm{GL}(V)$, a subspace $W$ of $V$ is $G$-\emph{invariant} if $\rho(g)({w}) \in W$ for every $g\in G$ and ${w} \in W$.\\
A representation is \emph{irreducible} if it contains no nonzero proper $G$-invariant subspaces.

If $V$ has finite dimension, the \emph{character} of $\rho$ is defined by the trace as the map
\begin{align*}
\chi_\rho\colon G &\longrightarrow K, \\
g &\longmapsto \mathrm{Tr}(\rho(g)).
\end{align*}
\end{defi}

\begin{lem}[Schur's lemma] \label{lem:schur}
The following hold:
\begin{enumerate}
\item If $\rho_1\colon G \longrightarrow \mathrm{GL}(V_1)$ and $\rho_2\colon G \longrightarrow \mathrm{GL}(V_2)$ are irreducible representations of the group $G$, then every nonzero homomorphism $\varphi$ between $\rho_1$ and $\rho_2$ is an isomorphism.\label{lem:schur1}
\item Let $\rho\colon G \longrightarrow \mathrm{GL}(V)$ be an irreducible representation of the group $G$ over an algebraically closed field $K$ and $\varphi\colon V \longrightarrow V$ a nonzero endomorphism of $\rho$.
Then $\varphi = \lambda \mathrm{Id}_V$ for some $\lambda \in K^\times$.
\label{lem:schur2}
\end{enumerate}
\end{lem}

\begin{cor}
Each irreducible complex representation of an abelian group is 1-dimensional.
\end{cor}

Thus, if $G$ is abelian, a complex representation $\rho\colon G \longrightarrow \mathrm{GL}(V)$ and the corresponding character $\chi_\rho\colon G \longrightarrow \mathbb{C}$ are the same function.

In addition, the set of all characters of $G$ with the operation of pointwise multiplication is a group $\hat{G}$ called the \emph{dual group} of $G$.
If $H \leq G$, then the subgroup $H^\perp = \{\chi_\rho \,|\, \chi_\rho(h) = 1,\, \forall h \in H\} \leq \hat{G}$ can be defined.

If $G$ is finite, $|H^\perp| = |G| / |H|$ and the fundamental theorem of finite abelian groups states that $G \cong \mathbb{Z}_{N_1} \times \mathbb{Z}_{N_2} \times \ldots \times \mathbb{Z}_{N_k}$, as explained in \cref{th:FTFAG}.
Furthermore, $|\chi_\rho(g)| = 1$ for each $g \in G$, since $|\chi_\rho(g^k)| = |\chi_\rho(g)|^k$ for all integers $k$.

\subsection{Polynomial-time reductions} 
\label{ssec:poly}

In this section, we introduce some basic concepts of computational complexity which are useful for the following discussion.
For a more thorough coverage of the topic, see \cite{AB09}.

\begin{defi}\label{def:bigoh}
    Let $f, g\colon \mathbb{R} \to \mathbb{R}$ be two functions.
    We write $f = O(g)$ if there exists a positive constant $c$ such that $f(x) \le c \cdot g(x)$ for every \emph{sufficiently large} $n$, i.e., if there is an integer $N$ such that, for every $n \ge N$, the previous inequality holds.
    Moreover, we write $f = \Omega(g)$ if $g = O(f)$, and we write $f = \Theta(g)$ if both $f = O(g)$ and $g = O(f)$ hold.
\end{defi}

To evaluate algorithm efficiency, we treat its running time as a function from the natural numbers to themselves that, given the input size of the algorithm, returns the number of operations (or steps) needed to compute its output.

\begin{defi}
    The \emph{running time} or \emph{complexity} of an algorithm is the number of the required operations, in function of the size of the input.
    An algorithm which takes an input of size $n$ (in bits) runs in \emph{polynomial time} if it performs $O(n^c)$ operations, for some positive constant $c$.
    An algorithm taking a size $n$ input runs in \emph{exponential time}
    if it ends after $O(c^n)$ steps for a constant $c>1$.
    A size $n$ input algorithm runs in \emph{logarithmic time} if it takes $O(\log n)$ to end, in \emph{polylogarithmic time} if its complexity is $O((\log n)^c)$, for some positive constant $c$.
    Those algorithm with running time bigger than polynomial, but still smaller than exponential, are called \emph{subexponential-time} algorithms.
\end{defi}

In cryptography, we refer to polynomial-time algorithms as \emph{efficient}.

We introduce a (partial) order relation among problems in order to compare their hardness.

\begin{defi}
    A problem $A$ \emph{polynomial-time reduces} to another problem $B$, and we denote it by $A \le_p B$, if there exists an algorithm $\mathcal{A}$ for solving $A$ which uses as subroutine an algorithm $\mathcal{B}$ for solving $B$ and runs in polynomial time if the subroutine does.
    This subroutine is also called \emph{oracle} for $B$.\\
    Two problems are \emph{computationally equivalent} if both polynomial-time reduce to the other.
\end{defi}

A polynomial-time reduction states that efficiently solving problem $B$ implies efficiently solving problem $A$, i.e., problem $B$ is harder than $A$ or, better, $A$ is at most as hard as $B$.

\subsection{Quantum computing} 
\label{ssec:quantum}

A quantum computer is a computing device that exploits the principles of quantum physics, allowing the implementation of quantum algorithms that seem to be able to solve in polynomial time certain hard problems currently intractable by classical algorithms.
In this section, we recall its basic notions exploited in the following.
For further details, see \cite{NC10}.

\begin{defi}
The building block used to store information is the quantum bit, or \textit{qubit} in short, a quantum analog of the classical bit.
It can assume two possible states, denoted by using the ket notation as $\ket{0}$ and $\ket{1}$, but thanks to the fundamental principle of quantum superposition, it can also assume multiple states simultaneously.
Thus, the state of a qubit is described as a linear combination
\[
\ket{\alpha} = a_0 \ket{0} + a_1 \ket{1}
\]
where $a_i \in \mathbb{C}$ is the \emph{probability amplitude} of the output $i$.
This is related to the \textit{measurement operator}, i.e., the fundamental feature that allows to obtain a numerical result from a qubit.
The output of this operation applied to $\ket{\alpha}$ is:
\begin{itemize}
    \item $0$ with probability $|a_0|^2$;
    \item $1$ with probability $|a_1|^2$.
\end{itemize}
Since the total probability has to be $1$, there is a constraint on the amplitudes given by $|a_0|^2 + |a_1|^2 = 1$.
\end{defi}

\begin{rem}
As a direct consequence of the given definition, all the qubits can be represented as normalized complex vectors and the obtained space is
\[
\left\{\begin{pmatrix}
     a_0 \\ a_1 
\end{pmatrix} \in \mathbb{C}^2 : |a_0|^2 + |a_1|^2 = 1
\right\}.
\]
Geometrically, this is a sphere called \textit{Bloch sphere}.

In particular, $\ket{0} = \binom{1}{0}$ and $\ket{1} = \binom{0}{1}$ form the so-called \textit{computational basis}.
\end{rem}

\begin{defi}
A \textit{register} is a set of $n$ qubits combined using the \textit{tensor} (or \textit{Kronecker}) \textit{product} $\otimes$.
The resulting representation is a complex vector of size $2^n$ containing all the obtained amplitudes.
The tensor product between orthonormal basis vectors returns the $2^n$ vectors of the computational basis, i.e., for every $i = (i_{n-1}, \ldots, i_0) \in \mathbb{Z}_2^n$, we have
\[
\ket{i} = \ket{i_{n-1} \ldots i_0} = \ket{i_{n-1}} \otimes \ldots \otimes \ket{i_0}.
\]
As an alternative notation, orthonormal basis vectors can be written as the integers modulo $2^n$ that they represent as a binary array $i = \sum_{j = 0}^{n - 1} 2^j i_j \in \mathbb{Z}_{2^n}$.
\end{defi}

\begin{rem}
For example, when $n = 2$, the state resulting from the combination of $\ket{\alpha} = a_0 \ket{0} + a_1 \ket{1}$ and $\ket{\beta} = b_0 \ket{0} + b_1 \ket{1}$ is
\[
\ket{\alpha} \otimes \ket{\beta} = a_0 b_0 \ket{00} + a_0 b_1 \ket{01} + a_1 b_0 \ket{10} + a_1 b_1 \ket{11}.
\]
where $\sum_{i,j = 0}^1 |a_i b_j|^2 = 1$.
Here, the new orthonormal basis vectors are 
\begin{gather*}
    \ket{00} = \ket{0}\otimes\ket{0} = \ket{0}, \\
    \ket{01} = \ket{0}\otimes\ket{1} = \ket{1}, \\
    \ket{10} = \ket{1}\otimes\ket{0} = \ket{2}, \\
    \ket{11} = \ket{1}\otimes\ket{1} = \ket{3}.
\end{gather*}

Generally, if at least one of the states can assume only one possible value, then the tensor product can be omitted.
\end{rem}
\begin{rem}
We extend the use of notation $\ket{\cdot}$ in the following way: let $x$ be any mathematical object that can be stored in a computer and assume that we fix a binary codification $B(x)$ for $x$, then $\ket{x}:=\ket{B(x)}$.
\end{rem}

\begin{defi}
A \emph{unitary operator} is a complex square matrix $U$ such that $U U^* = U^* U = I$, where $U^*$ denotes the conjugate transpose of the matrix $U$ (also called the Hermitian adjoint in quantum mechanics, denoted as $U^\dagger$).
\end{defi}

\begin{rem}
Since qubits must maintain their unitary norm, in quantum algorithms only functions that can be converted in unitary operators can be used.
Given a non-quantum function $f$, its representation as a unitary operator suited for quantum computing (if exists) is usually denoted as $U_f$.
A common strategy to construct $U_f$ such that it is always reversible is to allow two registers in input: the first one will be the actual input of $f$ while the second one will contain its output, so that applying $U_f$ results in a superposition on input and output.
\end{rem}

\begin{defi}
The building blocks of quantum computing are the \textit{quantum logic gates}, which are the quantum analogous to the classical logic gates.
Here are some notable examples:
\begin{itemize}
\item Identity gate $\mathcal I\colon\ket{\alpha} \longmapsto \ket{\alpha}$, with matrix representation on a single qubit
\[
\mathcal I = \begin{pmatrix}
    1 & 0 \\
    0 & 1
\end{pmatrix}.
\]
\item Pauli gates $\mathcal X, \mathcal Y, \mathcal Z$ represent rotations by $\pi$ radians of a qubit around the $x, y, z$ axes of the Bloch sphere, respectively.
The gate $\mathcal X$ is the quantum analogous of the NOT gate.
Their matrix representations are
\[
\mathcal X = \begin{pmatrix}
    0 & 1 \\
    1 & 0
\end{pmatrix} = \mathcal{NOT}, \qquad
\mathcal Y = \begin{pmatrix}
    0 & -\im \\
    \im & 0
\end{pmatrix}, \qquad
\mathcal Z = \begin{pmatrix}
    1 & 0 \\
    0 & -1
\end{pmatrix}.
\]
\item Hadamard gate allows to create the uniform superposition, i.e., $50\%$ output $0$ and $50\%$ output $1$.
In particular
\[
\mathcal H\colon\ket{0} \longmapsto \frac{1}{\sqrt{2}}(\ket{0}+\ket{1}), \quad
\ket{1} \longmapsto \frac{1}{\sqrt{2}}(\ket{0}-\ket{1}),
\]
and its matrix representation is
\[
\mathcal H = \frac{1}{\sqrt{2}}\begin{pmatrix}
    1 & 1 \\
    1 & -1
\end{pmatrix}.
\]
\item Phase shift gates are rotations about the $z$-axis on the Bloch sphere by $\varphi$ radians, i.e., 
\[
\mathcal P_\varphi\colon\ket{0} \longmapsto \ket{0}, \quad
\ket{1} \longmapsto \ex{\varphi \im} \ket{1}.
\]
They do not change the probability of the measurement but only the phase of the quantum state.
The matrix representation is
\[
\mathcal P_\varphi = \begin{pmatrix}
    1 & 0 \\
    0 &  \ex{\varphi \im}
\end{pmatrix}.
\]
A useful particular case is when $\varphi = 2 \pi / 2^k$, that gives the dyadic rational phase shift gate, denoted as
\[
\mathcal R_k = \begin{pmatrix}
    1 & 0 \\
    0 &  \ex{\frac{2 \pi}{2^k} \im}
\end{pmatrix}.
\]
\item Swap gate takes two qubits and swaps their values, it is given by
\[
\mathcal{SWAP} = \begin{pmatrix}
1 & 0 & 0 & 0 \\
0 & 0 & 1 & 0 \\
0 & 1 & 0 & 0 \\
0 & 0 & 0 & 1
\end{pmatrix}.
\]
\item Each gate can be \textit{controlled}, meaning that one or more qubits act as enabler of the operation.
Specifically, if an operator $\mathcal U$ acts on one qubit and is controlled by another qubit, then it is applied only when the control qubit is $1$, where usually the first qubit is the control qubit.
In formulas the controlled $\mathcal U$, denoted by $\mathcal{CU}$, acts as
\begin{align*}
& \ket{00} \longmapsto \ket{00}, \\
& \ket{01} \longmapsto \ket{01}, \\
& \ket{10} \longmapsto \ket{1} \otimes \mathcal U \ket{0} = \ket{1} (u_{00} \ket{0} + u_{10} \ket{1}), \\
& \ket{11} \longmapsto \ket{1} \otimes \mathcal U \ket{1} = \ket{1} (u_{01} \ket{0} + u_{11} \ket{1}),
\end{align*}
and in matrix notation
\[
\mathcal{CU} = \begin{pmatrix}
1 & 0 & 0 & 0 \\
0 & 1 & 0 & 0 \\
0 & 0 & u_{00} & u_{01} \\
0 & 0 & u_{10} & u_{11}
\end{pmatrix}.
\]
\item Measurement $\mu$ allows to obtain a classical binary output from a quantum state, collapsing the superposition to a single value depending on the quantum probabilities.
Since the input is impossible to retrieve, this is an irreversible operation and hence it is not a quantum gate.
\end{itemize}
\end{defi}

\begin{defi} \label{def:dft}
Let $\omega = \ex{\frac{2 \pi \im}{N}}$ be a primitive $N$-th root of unity, the \emph{Discrete Fourier Transform} (DFT) is
\begin{equation*}
\mathbb{C}^N \ni \begin{pmatrix}
v_0 \\
v_1 \\
v_2 \\
\vdots \\
v_{N - 1}
\end{pmatrix}=v \mapsto F_N v 
= \frac{1}{\sqrt{N}} \begin{pmatrix}
1 & 1 & 1 & \ldots & 1 \\
1 & \omega & \omega^2 & \ldots & \omega^{N - 1} \\
1 & \omega^2 & \omega^4 & \ldots & \omega^{2 (N - 1)} \\
\vdots & \vdots & \vdots & \ddots & \vdots \\
1 & \omega^{N - 1} & \omega^{2 (N - 1)} & \ldots & \omega^{(N - 1)^2}
\end{pmatrix} 
\begin{pmatrix}
v_0 \\
v_1 \\
v_2 \\
\vdots \\
v_{N - 1}
\end{pmatrix}
,\end{equation*}
so that the $j$-th component of the resulting vector is
\begin{equation*}
    \frac{1}{\sqrt{N}} \sum_{k = 0}^{N - 1} \omega^{j k} v_k, 
    \qquad \text{for }0 \leq j < N.
\end{equation*}
\end{defi}

By simply performing the matrix product, it is easy to prove that the matrix in the DFT is unitary.
Consequently, the inverse of the DFT of $v \in \mathbb{C}^N$ is simply the vector of components
\begin{equation*}
    \frac{1}{\sqrt{N}} \sum_{k = 0}^{N - 1} \omega^{- j k} v_k,
    \qquad \text{for } 0 \leq j < N.
\end{equation*}

\begin{rem}
Despite the conventional definition of the DFT has negative exponents, i.e., is the formula introduced here as its inverse, this work defines the DFT with positive exponents in order to be coherent with its quantum alternative.
\end{rem}

\begin{defi} \label{def:qft}
Let $N = 2^n$ and $\omega = \ex{\frac{2 \pi \im}{N}}$ be a primitive $N$-th root of unity.
The \emph{Quantum Fourier Transform} (QFT) of a quantum state $\ket{\alpha} = \sum_{k = 0}^{N-1} a_k \ket{k}$ is defined as
\begin{equation*}
    \mathcal F_N \colon\ket{\alpha} \longmapsto F_N \ket{\alpha},
\end{equation*}
so that the components of the resulting state are
\begin{equation*}
    \frac{1}{\sqrt{N}} \sum_{k = 0}^{N - 1} \omega^{j k} a_k,
    \qquad \text{for } 0 \le j < N.
\end{equation*}
\end{defi}

Again, the inverse of the QFT of $\ket{\alpha}$ is simply the state of the components
\begin{equation*}
    \frac{1}{\sqrt{N}} \sum_{k = 0}^{N - 1} \omega^{- j k} a_k,
    \qquad \text{for } 0 \le j < N.
\end{equation*}

\begin{rem}
The QFT can be implemented using $O(n^2)$ quantum gates of only three types, namely the Hadamard gate, the controlled dyadic rational phase shift gate and the swap gate.
\end{rem}

\subsection{Hidden Subgroup Problem}
\label{ssec:hsp}

Here we recall the definition of the Hidden Subgroup Problem, which has very important consequences in cryptanalysis since it can be related to many important problems, including those at the basis of the security of public-key cryptosystems.
The following sections address different contexts, considering the two main cases in which the group is either abelian or non-abelian.
In both cases, we examine the most relevant mathematical problems, from a cryptographic point of view, that can be reduced to an instance of the Hidden Subgroup Problem.

\begin{defi}
A function $f$ defined on a group $G$ \emph{hides} $H \leq G$ when
\begin{equation*}
f(g') = f(g) 
\qquad \text{if and only if}\qquad g' \in gH;
\end{equation*}
i.e., $f$ is constant on each coset of $H$ but gives distinct values on distinct cosets.
\end{defi}

\begin{defi}
Given a group $G$, a set $S$, and a function $f\colon G \longrightarrow S$ that hides $H \leq G$, the \emph{Hidden Subgroup Problem} (HSP) consists in finding a set of generators of the subgroup $H$.
\end{defi}

\section{Abelian HSP}
\label{sec:abel}

When the group $G$ is finite and abelian, there is a \emph{standard method} to solve the HSP, firstly introduced in 1995 in an unpublished work \cite{Kit95} by Kitaev.
This quantum routine generalizes the well known Shor's algorithm \cite{Sho94} that allows quantum computers to efficiently solve the main problems behind classical public-key cryptosystems.
In particular, the procedure applies the QFT to the group and, to do so, it exploits results from representation theory.

As observed in \cref{th:FTFAG}, a finite abelian group is isomorphic to $\mathbb{Z}_{N_1} \times \ldots \times \mathbb{Z}_{N_\ell}$.

When $\ell = 1$, we have that $G \cong \mathbb{Z}_N$ and the QFT is defined by a matrix of size $N \times N$.
If we do not consider the normalizing factor $1 / \sqrt{N}$, its $j$-th row can be written as
\begin{align*}
\chi_j\colon \mathbb{Z}_N &\longrightarrow \mathbb{C}, \\
k &\longmapsto \omega^{j k}.
\end{align*}
Since $\chi_j(h + k) = \chi_j(h) \chi_j(k)$, these are 1-dimensional representations of $\mathbb{Z}_N$ and, in particular, they are all the characters of $\mathbb{Z}_N$, giving the dual group $\hat{\mathbb{Z}}_N$.

Thus, the QFT maps the standard basis of $\mathbb{C}^N$ into the orthonormal basis corresponding to the characters, in formulas
\begin{equation*}
\mathcal F_N\colon \ket{j} \longmapsto \ket{\chi_j} 
= \frac{1}{\sqrt{N}} \sum_{k = 0}^{N - 1} \chi_j(k) \ket{k} 
= \frac{1}{\sqrt{N}} \sum_{k = 0}^{N - 1} \omega^{j k} \ket{k}.
\end{equation*}

\begin{rem}
\label{rem:qft}
In \cref{def:qft}, the QFT is introduced as an operator of size \mbox{$N = 2^n$}, while here it is applied to a generic group with arbitrary size.
Until this size is known, there is an approximation of the QFT that can be applied to odd-dimensional spaces, see \cite{Lom07}.
In the following, quantum operators are applied to groups of arbitrary size omitting this approximation for simplicity.
A different problem arises when the size is unknown (e.g., \cref{ssec:ofp}), in which case adequate solutions should be adopted.
\end{rem}

When $\ell > 1$, the $|G| = N_1 N_2 \ldots N_m$ characters are products of the characters of the cyclic groups $\mathbb{Z}_{N_i}$.
Since the characters are pairwise orthogonal, the QFT corresponding to $G$ is the tensor product of the QFTs for the individual cyclic groups.
For example, the QFT corresponding to $\mathbb{Z}_2$ is the Hadamard gate $\mathcal H$, so the QFT corresponding to $\mathbb{Z}_2^n$ is $\mathcal H^{\otimes n}$ (which is very different from the QFT corresponding to $\mathbb{Z}_{2^n}$).

\begin{prop}\cite{Kit95}
\label{th:kitaev}
Every instance of HSP over a finite abelian group $G$ with an efficiently-implementable function $f \colon G \longrightarrow S$ that hides $H \leq G$ can be solved by a polynomial quantum algorithm.

In particular, if $n = |G|$ and $m = |S|$, a set of generators of the hidden subgroup $H$ can be found with error probability less than $1/3$ by classically processing $O(n m \log n m)$ outputs obtained after the following quantum steps.
\begin{algo}\label{algo:AHSP}\hfill\vspace{-0.3cm}
\begin{enumerate}
\item Initiate a register (or more) representing elements in $G$ and an auxiliary register (or more) for elements in $S$, respectively, both in the state $\ket{0}$.

\item Apply Hadamard on the group register, i.e.,
\begin{equation*}
\left( \mathcal H^{G} \otimes \mathcal I^{S} \right) (\ket{0}^{\otimes n + m}) 
= \frac{1}{\sqrt{|G|}} \sum_{g \in G} \ket{g} \ket{0}^{\otimes m}.
\end{equation*}

\item Compute the unitary operator $U_f$, representing the black-box function $f$, that uses the group register as input and saves the output in the auxiliary register, thus obtaining

\begin{equation*}
\mathcal U_f \left( \frac{1}{\sqrt{|G|}} \sum_{g \in G} \ket{g} \ket{0}^{\otimes m} \right) = \frac{1}{\sqrt{|G|}} \sum_{g \in G} \ket{g} \ket{f(g)}.
\end{equation*}

\item Measure the auxiliary register, so that it assumes some value $f(\bar{g})$ for an unknown $\bar{g} \in G$, while the group register collapses to a superposition over the $g \in G$ with the same $f$-value as $\bar{g}$
\begin{equation*}
\left(\mathcal I^{G} \otimes \mu^{S} \right) \left( \frac{1}{\sqrt{|G|}} \sum_{g \in G} \ket{g} \ket{f(g)} \right) 
= \frac{1}{\sqrt{|H|}} \sum_{h \in H} \ket{\bar{g} \cdot h}.
\end{equation*}

\item Apply the QFT corresponding to $G$ to the obtained state, which returns
\begin{equation*}
\mathcal F_{|G|} \left( \frac{1}{\sqrt{|H|}} \sum_{h \in H} \ket{\bar{g} \cdot h} \right) 
= \frac{1}{\sqrt{|H|}} \sum_{h \in H} \ket{\chi_{\bar{g} \cdot h}}.
\end{equation*}
\item Measure and output the result.
\end{enumerate}
\end{algo}
\end{prop}

\begin{rem}
Here are some considerations on \cref{algo:AHSP} above:
\begin{itemize}
\item Step 2 generates in the group register a uniform superposition over the elements of $G$.

\item In Step 3 the uniform superposition is extended to the auxiliary register, so that it contains the image through $f$ of the $g$ in the group register.

\item After Step 4, the state is a uniform superposition over all the elements in the coset $\bar{g} H$ of $H$.

\item Step 5 returns a superposition over the labels of $H^\perp$ given by
\begin{align*}
\frac{1}{\sqrt{|H|}} \sum_{h \in H} \ket{\chi_{\bar{g} \cdot h}} 
&= \frac{1}{\sqrt{|H| |G|}} \sum_{h \in H} \sum_{g \in G} \chi_{\bar{g} \cdot h}(g) \ket{g} \\
&= \frac{1}{\sqrt{|H| |G|}} \sum_{g \in G} \chi_{\bar{g}}(g) \sum_{h \in H} \chi_h(g) \ket{g}.
\end{align*}
Since, for each $g \in G$, the character $\chi_g$ restricted to $H$ is a character of $H$, and because it is constant of value 1 if $\chi_g \in H^\perp$, the following equation holds
\begin{equation}\label{eq:char}
\sum_{h \in H} \chi_h(g) = \sum_{h \in H} \chi_g(h) = \begin{cases}
|H| & \text{if } \chi_g \in H^\perp, \\
0 & \text{if } \chi_g \notin H^\perp,
\end{cases}
\end{equation}
and the resulting state is
\begin{equation*}
\sqrt{\frac{|H|}{|G|}} \sum_{g \,|\,\chi_g \in H^\perp} \chi_g(\bar{g}) \ket{g}.
\end{equation*}
The phases (i.e., the coefficients of the quantum states) do not affect the probabilities of the measurement because $|\chi_g(\bar{g})|^2 = 1$.
Thus, the obtained superposition is uniform on the labels of elements of $H^\perp$.

\item Each sampled element $g \in G$ gives a constraint on $H$ because $\chi_g(h) = 1$ for each $h \in H$.
Hence, information to find the generators of $H$ itself can be obtained after \emph{ad-hoc} computations on a sufficient number of samples, depending on the specific instance of HSP.
\end{itemize}
\end{rem}

\begin{defi}
The method described above is known as \emph{standard method} and the quantum routine on which it relies is called \emph{Fourier sampling}.
The first four steps of the routine form a process called \emph{coset sampling} because the final state is a superposition over all the elements in a coset of $H$.
\end{defi}

In 2017, an unpublished work by Gogioso and Kissinger \cite{GK17} gave an interesting and detailed graphical diagrammatic proof of the correctness for the abelian HSP algorithm.

New developments on abelian HSP were obtained in \cite{II22}, where the authors introduce a polynomial time exact quantum algorithm for $\mathbb{Z}_{m^k}^n$.
They also present applications to compute the structure of abelian and solvable groups whose order has the same prime factors as $m$.

In the following, the main mathematical problems that can be reduced to the abelian HSP are described.

\subsection{Simon's problem} 
\label{ssec:simon}

Simon's problem is one of the first examples of problems that can be solved in polynomial time by a quantum algorithm \cite{Sim94}, while classical algorithms are exponential.
Here we introduce the formulation of the problem and we highlight how it can be seen as an instance of HSP with an abelian group.

\begin{defi}
Let $m \ge n$ be positive integers and $f \colon \mathbb{Z}_2^n \longrightarrow \mathbb{Z}_2^m$ be a function satisfying the following property:
\emph{there is $x \in \mathbb{Z}_2^n$ such that, for each $g \neq g' \in \mathbb{Z}_2^n$,}
\begin{equation*}
f(g) = f(g') \quad\text{if and only if}\quad g' = g \oplus x,
\end{equation*}
where $\oplus$ is the bitwise exclusive or (XOR).
The \textit{Simon's problem} consists of finding the bit string $x$.
\end{defi}

It is easy to formulate the Simon's problem as an instance of HSP with the group $(\mathbb{Z}_2^n, \oplus)$.
Since for every $g, g' \in \mathbb{Z}_2^n$ it is true that
\begin{equation*}
f(g) = f(g') 
\quad\text{if and only if}\quad g' \in \{g, g \oplus x\} 
\quad\text{if and only if}\quad g' \in g \oplus \{0, x\},
\end{equation*}
the function $f$ hides the subgroup $H = \{0, x\}$ of $\mathbb{Z}_2^n$.
In particular, finding $H$ means solving the Simon's problem.

In \cref{algo:AHSP}, the QFT corresponding to $\mathbb{Z}_2^n$ is simply $\mathcal H^{\otimes n}$ and the $2^n$ characters are $\chi_{g}(g') = (-1)^{g \cdot g'}$, where 
\begin{equation*}
g \cdot g' := g_1 g'_1 \oplus \ldots \oplus g_n g'_n,\quad\text{for }g=(g_1,\ldots,g_n),g'=(g'_1,\ldots,g'_n)\in\mathbb Z_2^n.
\end{equation*}

The algorithm then uniformly samples from labels of elements of
\begin{equation*}
H^\perp 
= \{\chi_g \,|\, \chi_g(h) = (-1)^{g \cdot h} = 1, \;  h \in H\} 
= \{\chi_g \,|\, g \cdot x \equiv 0 \text{ (mod } 2)\}.
\end{equation*}
After $O(n)$ samples, the $n - 1$ linearly independent equations $g \cdot x \equiv 0 \text{ (mod } 2)$ can be solved using Gaussian elimination in order to find $x$.

Recently, a strand of scientific literature aims at finding better quantum strategies for solving Simon's problem and its generalizations, in order to decrease the number of qubits required by algorithms \cite{KNP07, TXQ22, LQL24, SKZ25}.

\subsection{Order finding and integer factorization}
\label{ssec:ofp}

The order finding problem is one of the crucial hard problems on which the security of several cryptosystems is based.
Indeed, it is well known that the factors of an integer $N$ can be efficiently computed if the order of the elements in $\mathbb Z_N^\times$ can be found.
We start recalling how the integer factorization problem can be easily reduced to the order finding problem.
Then, we show that the order finding problem can be seen as a HSP with an abelian group and consequently solved in polynomial time by quantum computers.

\begin{defi}
Given an element $a$ of a finite group $(A, \cdot)$, finding its order, i.e., the smallest positive integer $x$ such that
\begin{equation*}
 a^x = 1,
\end{equation*}
is the \textit{Order Finding Problem} (OFP).
\end{defi}

\begin{defi}
The \textit{Integer Factorization Problem} (IFP) consists in finding the prime factors of a given integer $N$.
\end{defi}

Excluding the particular instance where $N$ is the power of an odd prime, the IFP problem can be reduced to the OFP through the following classical routine.
\begin{algo}\hfill\vspace{-0.3cm}
\begin{enumerate}
\item Pick a random integer $1 < a < N$.
\item If $a$ is not coprime with $N$, then $\gcd(a, N)$ is a non-trivial factor of $N$ and the problem is solved (if necessary, repeat the procedure for other factors).
\item Otherwise, find the order $x \le \varphi(N)$ of $a \in \mathbb{Z}_N^\times$, where $\varphi(N)$ is the Euler totient function, i.e., solve the OFP.
\item If $x$ is odd, then return to Step~1.
\item If $x$ is even, then $a^{\frac{x}{2}} \not\equiv 1 \pmod{N}$, for otherwise the order is less than $\tfrac{x}{2}$.
\item If $a^{\frac{x}{2}} \equiv -1 \pmod{N}$, then return to Step~1.
\item Otherwise, $N$ divides $a^x - 1 = (a^{\frac{x}{2}} + 1) (a^{\frac{x}{2}} - 1)$ and none of the factors is a multiple of $N$.
Thus, a non-trivial factor of $N$ is provided by $\gcd(a^{\frac{x}{2}} + 1, N)$ or $\gcd(a^{\frac{x}{2}} - 1, N)$ and the problem is solved (if necessary, repeat the procedure for other factors).
\end{enumerate}
\end{algo}

If $N$ factors as $pq$, with $p$ and $q$ primes, which is the most useful case in cryptography, the probability of returning at Step~1 is at most $1 / 2$ since:
\begin{itemize}
\item By Chinese Remainder Theorem, choosing a uniformly random $a \in \mathbb{Z}_N$ is equivalent to pick $a_p \in \mathbb{Z}_p$ and $a_q \in \mathbb{Z}_q$ independently.
If $x_p$ and $x_q$ are their respective orders, then $(a_p, 1)$ generates the subgroup of size $x_p$ of the group of size $x$ generated by $(a_p, a_q)$, and analogously for $(1, a_q)$.
Thus, Lagrange's theorem implies that $x_p$ and $x_q$ divide $x$, which means that if $x$ is odd, then $x_p$ and $x_q$ must be odd.
The probability that $x_p$ (or $x_q$) is odd is $1 / 2$ because $\mathbb{Z}_p^\times$ is cyclic of even size, so half of its elements are squares.
Hence, the probability that $x$ is odd is at most $(1 / 2)^2 = 1 / 4$.
\item If $x$ is even and $a^{\frac{x}{2}} \equiv -1 \pmod{N}$, then $a^{\frac{x}{2}} \equiv -1 \pmod{p}$ and $a^{\frac{x}{2}} \equiv -1 \pmod{q}$.
This event has probability at most $(1 / 2)^2 = 1 / 4$.
\end{itemize}
Collecting the two results allows to obtain
\begin{equation*}
\mathrm{Pr}[\text{return to Step 1}] = \mathrm{Pr}[x \text{ odd}] + \mathrm{Pr}[a^{\frac{x}{2}} \equiv -1] \leq \frac{1}{4} + \frac{1}{4} = \frac{1}{2}.
\end{equation*}

Through the described classical reduction, the integer factorization can be solved by finding the order $x$ of the sampled integer $a \in \mathbb{Z}_N^\times$.

Now we see that it is possible to formulate the OFP as an instance of HSP.
Consider the group $(\mathbb{Z}_{|A|}, +)$ and the function $f\colon \mathbb{Z}_{|A|} \longrightarrow A$ defined by $f(g) = a^g$, for a fixed $a \in \mathbb Z_N^\times$.
Since for $g, g' \in \mathbb{Z}_{|A|}$
\begin{equation*}
f(g') = f(g) 
\quad\text{if and only if}\quad g' = g + j x 
\quad\text{if and only if}\quad g' \in g + \langle x \rangle,
\end{equation*}
$f$ hides the subgroup $H = \langle x \rangle = \{0, x, 2 x, \ldots, (|A| - 1) x \}$, where $x$ is the order of $a$.

In this case, the state at Step~4 of \cref{algo:AHSP} is
\begin{equation*}
\frac{1}{\sqrt{|A|}} \sum_{j = 0}^{|A| - 1} \ket{\bar{g} + j x},
\end{equation*}
and moving to Step~5, i.e., applying the QFT to the group register, results in
\begin{equation*}
\frac{1}{\sqrt{|A|}} \sum_{j = 0}^{|A| - 1} \frac{1}{\sqrt{|A|}} \sum_{g = 0}^{|A| - 1} e^{\frac{2 \pi \im}{|A|} (\bar{g} + j x) g} \ket{g}=\frac{1}{|A|} \sum_{g = 0}^{|A| - 1} e^{\frac{2 \pi \im}{|A|} \bar{g} g} \sum_{j = 0}^{|A| - 1} e^{\frac{2 \pi \im}{|A|} j x g} \ket{g}
\end{equation*}
where the probability of picking an element depends on the value of the inner sum that is
\begin{equation*}
\sum_{j = 0}^{|A| - 1} \left( e^{\frac{2 \pi \im}{|A|} g x} \right)^j= 
\begin{cases}
|A| & \text{if } g x \equiv 0 \bmod{|A|}, \\
0 & \text{otherwise.}
\end{cases}
\end{equation*}
Thus, the Fourier sampling picks labels from those in $H^\perp = \{\chi_g: gx/|A| \in \mathbb{Z}\}$, which means that the output is $g = k|A|/x$.
Consequently, $g$ is an integer multiple of $|A|/x$ and $k$ satisfies $0 \le k < x$, since $g<|A|$.
Since only $|A|$ and the measured $g$ are known, it is possible to obtain $g / |A| = k / x$.
By adjusting the size of the group register, the procedure can provide good approximations $\alpha$ of $k/x$  whose quality is given by
$$\left|\alpha - \cfrac{k}{x}\right| < \cfrac{1}{2x^2},$$
see \cite[pag.~229]{NC10}.
Thus, $k/x$ is a convergent of the continued fraction expansion of the sampled value $\alpha$, see \cite[Theorem 10.18]{Ros86}
Note that the convergents of a continued fraction are in lowest terms, and therefore the actual value of the order of $a$ could be a multiple of the denominator of a convergent.
Hence, all its multiples smaller than $\varphi(N)$ must be tested to find the value of $x$.
Overall, the number of denominators of the convergents to test is $\log |A|$, because computing the continued fraction expansion of a rational number is derived from the Euclidean algorithm.
Then, for each of them, the number of multiples to check is $\varphi(N)/|A|$.

\subsection{Discrete logarithm} 
\label{ssec:dlp}

The discrete logarithm problem is one of the main problems whose computational hardness is crucial for the security of many public-key cryptosystems and key exchange protocols.

\begin{defi}
Given a generator $a$ of a cyclic group $(A,\cdot)$ of size $N$ and $b \in A$, the \textit{Discrete Logarithm Problem} (DLP) is finding the unique $0 \le x < N$ such that $a^x = b$.
\end{defi}

The DLP is an instance of the HSP with the group $(\mathbb{Z}_N \times \mathbb{Z}_N, +)$ and the function $f\colon \mathbb{Z}_N \times \mathbb{Z}_N \longrightarrow A$ defined by
\begin{equation*}
f(g_1, g_2) := a^{g_1} b^{-g_2} = a^{g_1 - g_2 x}.
\end{equation*} 

In particular, for $g = (g_1, g_2), g' = (g'_1, g'_2) \in \mathbb{Z}_N \times \mathbb{Z}_N$ we have
\begin{equation*}
f(g') = f(g)
\quad \text{if and only if}\quad  g'_1 - g_1 \equiv (g'_2 - g_2) x \bmod N,
\end{equation*}
that, for a suitable $k \in \mathbb Z_N$, is equivalent to the system
\begin{equation*}
\begin{cases}
g_1' - g_1 \equiv kx\bmod N, \\
g_2' - g_2 \equiv k\bmod N.
\end{cases}
\end{equation*}
Hence,
\begin{equation*}
f(g') = f(g) 
\quad \text{if and only if}\quad g' \in g + \langle (x, 1) \rangle,
\end{equation*}
and $f$ hides the subgroup
\[
H = \langle (x, 1) \rangle = \{(0, 0),(x, 1), (2 x, 2), \ldots, ((N-1) x, N-1)\}.
\]
In the coset sampling in \cref{algo:AHSP}, $\mathbb{Z}_N\times\mathbb{Z}_N$ is modeled using two registers and, after measuring the auxiliary register, the state of the group registers at Step~4 is
\begin{equation*}
\frac{1}{\sqrt{N}} \sum_{k = 0}^{N - 1} \ket{\bar{g}_1 + k x} \ket{\bar{g}_2 + k}.
\end{equation*}
After the QFT corresponding to $\mathbb{Z}_N \times \mathbb{Z}_N$, the resulting state in Step 5 is
\begin{align*}
&\frac{1}{\sqrt{N}} \sum_{k = 0}^{N - 1} 
\left(\!\frac{1}{\sqrt{N}}\!\sum_{g_1 = 0}^{N - 1} \ex{\frac{2 \pi \im}{N} (\bar{g}_1 + k x) g_1}\!\ket{g_1}\!\right)\!
\left(\!\frac{1}{\sqrt{N}}\!\sum_{g_2 = 0}^{N - 1} \ex{\frac{2 \pi \im}{N} (\bar{g}_2 + k) g_2}\!\ket{g_2}\!\right) \\
&=\frac{1}{N \sqrt{N}} \sum_{g_1, g_2 = 0}^{N - 1} \ex{\frac{2 \pi \im}{N} (\bar{g}_1 g_1 + \bar{g}_2 g_2)} \sum_{k = 0}^{N - 1} \ex{\frac{2 \pi \im}{N} (g_1 x + g_2)}^k \ket{g_1} \ket{g_2},
\end{align*}
where the probability of picking an element depends on the value of the inner sum that is geometric over $k$.
Therefore,
\begin{equation*}
\sum_{k = 0}^{N - 1} \ex{\frac{2 \pi \im}{N} (g_1 x + g_2)}^k 
= \begin{cases}
N & \text{if } g_1 x + g_2 \equiv 0 \bmod{N}, \\
0 & \text{otherwise.}
\end{cases}
\end{equation*}
Thus, possible outputs have $g_2 = - g_1 x+jN$, with $j\in\mathbb Z$, that results in the state
\begin{equation*}
\frac{1}{\sqrt{N}} \sum_{g_1 = 0}^{N - 1} \ex{\frac{2 \pi \im}{N} (\bar{g}_1 - \bar{g}_2 x) g_1} \ket{g_1} \ket{- g_1 x},
\end{equation*}
and the Fourier sampling returns a label from those in $H^\perp = \{\chi_{(g_1, - g_1 x)} \,|\, g_1 \in \mathbb{Z}_N\}$, so that the solution of the DLP can be found directly from every sample with invertible first component.

\begin{rem}
    For the sake of completeness, we highlight that the DLP can also be seen as an instance of the OFP and, consequently, be solved with a strategy similar to the one in the previous section.
    In particular, the OFP can be recast as a \emph{period finding problem}, which asks to compute the period of a given periodic function, and the DLP can be solved by using the knowledge of the period of the function $f\colon \mathbb{Z}_N \times \mathbb{Z}_N \longrightarrow A$ defined above. Assume that $f$ is of period $(m,n)$, for some $m,n \in \mathbb{Z}_N$. Actually, we can also suppose that $m, n \in \mathbb Z^\times$, otherwise we have found a non-trivial factor of $N$. 
    Thus, we have 
    \[ f(g_1 + m, g_2 + n) = f(g_1, g_2), \]
    i.e.,
    \[a^{g_1+m-(g_2+n)x}=a^{g_1-g_2x}\]
    which means that
    \[a^{m+nx} = 1_A.\]
    Hence, we have 
    \[m+nx \equiv 0 \bmod{N}\]
    and we can recover the DLP solution from a solution of the OFP, since
    \[x \equiv - \frac{m}{n} \bmod{N}.\]
\end{rem}

More generally, the DLP is one of the most known examples of \emph{Semigroup Action Problem} (SAP), which, given in input a semigroup acting on a set and two elements of that set, consists in finding the semigroup element that map the first set element into the latter. 
In \cite{HPD24}, the authors give some sufficient conditions for the SAP to be reduced to an abelian HSP, such as the injectivity of the action considered over a module or a bound on the cardinality of the hidden subgroup least generating set with respect to that of the underlying ring.

\subsection{Discrete root} 
\label{ssec:root}

The discrete root problem, i.e., the extraction of a root in $\mathbb Z_N$, can be easily solved if the factorization of $N$ is known and therefore this problem can be efficiently solved by a quantum computer.

\begin{defi}
Let $N$ be a positive integer, $a \in \mathbb{Z}_N^\times$ such that $a \equiv x^e \pmod{N}$ for a given $e \in \mathbb{Z}_{\varphi(N)}$, the \textit{Discrete Root Problem} (DRP) is finding $x \in \mathbb{Z}_N^\times$.
\end{defi}

\begin{rem}
If $e$ is an invertible element of $ \mathbb{Z}_{\varphi(N)}$, then the DRP is the classic \textit{RSA problem}, i.e., retrieve the plaintext $M$ knowing the public key $(N, e)$ and the ciphertext $C \equiv M^e \pmod N$.
\end{rem}

Clearly, one way to reduce the DRP to the HSP is through the IFP, since knowing a factorization of $N$ makes finding the root $x$ feasible.
However, to the best of our knowledge, currently there is no direct reduction of the discrete root problem to an instance of the HSP, unlike the problems analyzed in the previous sections.
In the following, we propose some approaches to reduce this problem directly to an HSP instance, highlighting the issues that arise with each of the proposed attempts.

\paragraph{DLP-like approach.}
In an instance of RSA problem, the discrete root $x=a^y$ can be obtained by solving the DLP in $(\mathbb{Z}_{\varphi(N)}, +)$ given by 
\begin{equation*}
e y \equiv 1 \bmod{\varphi(N)}.
\end{equation*}
With this formulation, there is an instance of HSP with $G = (\mathbb{Z}_{\varphi(N)} \times \mathbb{Z}_{\varphi(N)}, +)$ and $f\colon G \longrightarrow \mathbb{Z}_{\varphi(N)}$ defined as
\[
f(g_1, g_2) := g_1 e - g_2 = e (g_1 - g_2 y).
\]
The hidden subgroup is obtained considering that  
\begin{equation*}
f(g') = f(g) 
\quad \text{if and only if}\quad g'_1 - g_1 \equiv (g'_2 - g_2) y \bmod \varphi(N),
\end{equation*}
for $g=(g_1,g_2),g'=(g_1',g_2')\in G$.
For a suitable $k\in\mathbb Z_{\varphi(N)}$, this is equivalent to the following system:
\begin{equation*}
\begin{cases}
g_1' - g_1 \equiv ky &\bmod\,\varphi(N), \\
g_2' - g_2 \equiv k &\bmod\,\varphi(N).
\end{cases}
\end{equation*}
Hence,
\begin{equation*}
f(g') = f(g) 
\quad\text{if and only if}\quad g' \in g + \langle (y, 1) \rangle,
\end{equation*}
and $f$ hides the subgroup $H = \langle (y, 1) \rangle$.

Despite this approach appearing promising for the DRP, a problem arises in the implementation of $f$: the reduction modulo $\varphi(N)$ is required, but $\varphi(N)$ is unknown since we do not have access to the factorization of $N$.

\paragraph{Mixed approach.}
Another idea is to work with the known modulo $N$, for instance, by taking the group $G = (\mathbb{Z}_N^\times \times \mathbb{Z}_{\varphi(N)}, *)$, where
\begin{equation*}
(g_1, g_2) * (g'_1, g'_2) = (g_1 g'_1, g_2 + g'_2),\qquad \text{for }g = (g_1, g_2),\; g' = (g'_1, g'_2) \in G,
\end{equation*}
and the function $f\colon G \longrightarrow \mathbb{Z}_N^\times$ given by
\[
f(g_1, g_2) := g_1^e a^{- g_2} = g_1^e x^{-e g_2},
\]
so that evaluating the function $f$ does not require knowing $\varphi(N)$.
In particular,
\begin{equation*}
\label{eq:mix}
f(g') = f(g) \quad\text{if and only if}\quad (g'_1 g_1^{-1})^e \equiv x^{e (g'_2 - g_2)}\bmod N.
\end{equation*}
Assuming that there are $r \ge 1$ distinct $e$-th roots $x_1, x_2, \ldots, x_r$, this is equivalent to the system
\begin{equation*}
\begin{cases}
g'_1 g_1^{-1} \equiv x_1^{k_1}\ldots x_r^{k_r} &\bmod\, N, \\
g_2' - g_2 \equiv k_1+\ldots+k_r &\bmod\,\varphi(N),
\end{cases}
\end{equation*}
for suitable $k_1,\ldots,k_r\in\mathbb Z_{\varphi(N)}$.
Hence,
\begin{equation*}
f(g')= f(g)\quad \text{if and only if}\quad g' \in g * H,
\end{equation*}
where 
$$
H = \langle (x_1, 1), \ldots, (x_r, 1) \rangle = \{(x_1^{k_1} \ldots x_r^{k_r}, k_1 + \ldots + k_r) \,|\, k_1,\ldots,k_r \in \mathbb{Z}_{\varphi(N)} \}.
$$

From the theoretical point of view, this approach fails when there are no $e$-th roots modulo $N$.
Moreover, the implementation of the QFT is efficient only when dealing with additive groups.
Thus, in order to use it to solve the described HSP, explicit isomorphism $\mathbb{Z}_N^\times \cong \prod_{p\vert N}\mathbb{Z}_{p^{k_p}}^\times \cong \prod_{p\vert N}\mathbb{Z}_{\varphi(p^{k_p})}$ and generators of the multiplicative cyclic groups are required, for which we need to know the prime factorization $N = \prod_{p\vert N} p^{k_p}$.

\paragraph{Exponent approach.}
In this approach, we formulate an instance of the HSP whose hidden subgroup hides the value of $\varphi(N)$.
Once this value is obtained, the DRP can be efficiently solved.
Actually, this approach does not directly solve the DRP but it could be interesting because, to the best of our knowledge, a direct reduction of the problem of computing $\varphi(N)$ (without knowing the factorization of $N$) to the HSP is not known.
Here, we use the group $G=(\mathbb{Z}_{\varphi(N)} \times \mathbb{Z}_{\varphi(N)}, +)$ and the function $f\colon \mathbb{Z}_{\varphi(N)} \times \mathbb{Z}_{\varphi(N)} \longrightarrow \mathbb{Z}_N^\times$ defined by 
$$
f(g, h) := a^{g + e h} \bmod{N}.
$$
In this case, denoting by $\text{ord}(a)$ the multiplicative order of $a$ in $\mathbb Z_N$, the hidden subgroup is obtained from
\begin{equation*}
f(g',h') = f(g,h) \quad \text{if and only if} \quad  g' - g \equiv - e (h' - h) \bmod{\text{ord}(a)},
\end{equation*}
that, for suitable $j,k\in\mathbb Z_{\varphi(N)}$, is equivalent to the system
\begin{equation*}
\begin{cases}
g' - g \equiv - k e + j \text{ord}\,a &\bmod\,\varphi(N), \\
h' - h \equiv k &\bmod\,\varphi(N).
\end{cases}
\end{equation*}
Hence,
\begin{equation*}
f(g',h')= f(g,h)\quad \text{if and only if}\quad  (g',h') \in (g,h) + \langle(- e, 1), (\text{ord}\, a, 0)\rangle.
\end{equation*}
This formulation allows to avoid all the problems that arise previously.
In \cref{algo:AHSP}, the group occupies two registers and the state in Step~4 is
\begin{equation*}
\frac{1}{\varphi(N)} \sum_{k, j = 0}^{\varphi(N) - 1} \ket{\bar{g} - k e + j \text{ord}\, a} \ket{\bar{h} + k},
\end{equation*}
so that the application of the two QFTs related to $\mathbb{Z}_{\varphi(N)}$ to the group registers gives
\begin{equation*}
\frac{1}{\varphi(N)^2} \sum_{k, j = 0}^{\varphi(N) - 1} 
\left( \sum_{g = 0}^{\varphi(N) - 1} e^{\frac{2 \pi \im}{\varphi(N)} (\bar{g} - k e + j\text{ord}\, a) g} \ket{g} \right)\left( \sum_{h = 0}^{\varphi(N) - 1} e^{\frac{2 \pi \im}{\varphi(N)} (\bar{h} + k) h} \ket{h} \right),
\end{equation*}
which can be simplified to
\begin{equation*}\frac{1}{\varphi(N)^2} 
\sum_{g, h = 0}^{\varphi(N) - 1} 
e^{\frac{2 \pi \im}{\varphi(N)} (\bar{g} g + \bar{h} h)}\sum_{k, j = 0}^{\varphi(N) - 1} e^{\frac{2 \pi \im}{\varphi(N)} (- keg + j\text{ord}(a)\,g + kh)} \ket{g} \ket{h}.
\end{equation*}
Here, the probability of measuring an element is higher if the exponential in the inner sum (over $k, j$) is equal to 1,
that is true if and only if
\begin{equation*}
\begin{cases}
g \equiv 0 &\bmod\,\frac{\varphi(N)}{\text{ord}\, a}, \\
h \equiv e g& \bmod\,\varphi(N).
\end{cases}
\end{equation*}
Thus, the labels from $H^\perp$ returned by the Fourier sampling satisfy this system of equations and can be used to find a multiple of $\varphi(N)$: given two samples $(g_1,h_1)$ and $(g_2,h_2)$, then 
\[
h_i \equiv e g_i \bmod\,\varphi(N),
\]
so that 
\[
\text{lcm}(g_1, g_2)\frac{g_1}{h_1} \equiv \text{lcm}(g_s, g_t)\frac{g_s}{h_s} \bmod \varphi(N),
\]
and consequently
\[
\left(\frac{h_2}{g_2} - \frac{h_1}{g_1}\right)\text{lcm}(g_2, g_1) \equiv 0 \bmod \varphi(N).
\]
Using this multiple of $\varphi(N)$, thanks to the results from \cite{Mil76} (which requires the extended Riemann hypothesis), the factorization of $N$ is obtainable in polynomial time.

\begin{rem}
Following \cref{rem:qft}, in the implementation of this approach, the group registers used as input for the QFT have unknown size $\varphi(N)$.
However, adopting an upper bound for $\varphi(N)$ (e.g., $N-\sqrt{N}$) should be enough to solve this limitation.
\end{rem}

\section{Non-abelian HSP}
\label{sec:nabel}

The HSP on non-abelian groups has collected lots of attentions because it presents many relations with theoretical and applicative topics. 
For instance, as we detail in this section, the non-abelian HSP can be linked to classical mathematical problems or to new security principles behind post-quantum cryptosystems, but also to the Grover's algorithm that allows to reduce the security of symmetric cryptography using a quantum computer \cite{LK07,KL07}.

However, the previous considerations with abelian groups do not hold when dealing with the finite non-abelian cases.
In particular, since the representations $\rho$ have dimension $d_\rho > 1$ and, hence, do not coincide with the corresponding characters $\chi_\rho$, the QFTs corresponding to non-abelian groups are more complicated than in the abelian case and a quantum algorithm to efficiently solve the non-abelian HSP is generally not known.

From a theoretical point of view, since $G$ is still finite, the number of its irreducible pairwise
non-isomorphic representations is finite, see \cite[Theorem 7, Section 2.4]{Ser96}.
More precisely, let $\hat{G}$ be the set of such representations, let $V_\rho$ be the vector space, and let $d_\rho$ denote the dimension of each $\rho \in \hat{G}$, then
$$
\frac{1}{|G|} \sum_{\rho \in \hat{G}} d_\rho \chi_\rho(g) = \begin{cases}
1 &\text{if }g=1, \\
0 &\text{if }g\ne 1,
\end{cases}
$$
and, since $\chi_\rho(1) = \text{Tr}(I_{d_\rho}) = d_\rho$, this implies that $|G| = \sum_{\rho \in \hat{G}} d_\rho^2$.

\begin{prop}
For each $\rho \in \hat{G}$ there is a basis $(v_1, \ldots, v_{d_\rho})$ of $V_\rho$ such that the matrix representation $\rho(g) \in \mathbb{C}^{d_\rho \times d_\rho}$ is unitary for each $g \in G$.
\end{prop}

This result means that, after a suitable basis is found, the QFT in the non-abelian case is still implementable as an efficient quantum circuit, despite a new definition is required.
In the following, $\rho(g)$ is the matrix representation in the basis given in the previous proposition.

\begin{defi}
The \emph{General Fourier Transform} (GFT) over $G$ is defined as
$$
\mathcal F_G\colon \ket{g} \longmapsto \frac{1}{\sqrt{|G|}} \sum_{\rho \in \hat{G}} \sqrt{d_\rho} \sum_{i, j = 1}^{d_\rho} \rho(g)_{i, j} \ket{\rho} \ket{i} \ket{j}.
$$
\end{defi}

\begin{prop}
The GFT $\mathcal F_G$ is a unitary transform.
\end{prop}

In the Fourier sampling in \cref{algo:AHSP}, the result of the coset sampling (Steps~1--4) remains unchanged, but the state of the group register after the QFT (Step~5) becomes
$$
\frac{1}{\sqrt{|G|}} \sum_{\rho \in \hat{G}} \sqrt{d_\rho} \sum_{i,j = 1}^{d_\rho} \left( \frac{1}{\sqrt{|H|}} \sum_{h \in H} \rho(\bar{g} h) \right)_{i,j} \ket{\rho} \ket{i} \ket{j}.
$$

\begin{defi}
It is possible to distinguish two forms of Fourier sampling:
\begin{itemize}
\item the \emph{weak Fourier sampling} returns a representation $\rho \in \hat{G}$ with probability $\text{Pr}[\rho] = \sum_{i, j = 1}^{d_\rho} \text{Pr}[\rho, i, j]$;
\item the \emph{strong Fourier sampling} returns a representation $\rho \in \hat{G}$ and the indices $i,j\in\{1,\ldots,d_\rho\}$ with probability $\text{Pr}[\rho, i, j]$.
\end{itemize}
\end{defi}

When $G$ is abelian, then $d_\rho = 1$ and the two definitions coincide.
More precisely, the samples are random elements $g_1, \ldots, g_s \in G$ such that
$$
\chi_{g_i} \in H^\perp
= \{\chi_g \,|\, \chi_g(H) = \{1\}\}
= \{\chi_g \,|\, H \subseteq \ker(\rho_g)\},
$$
and the systems of equations $\chi_{g_i}(h) = \chi_h(g_i) = 1$ can be solved to find $h \in H$.
In other words, the output of the algorithm is a random element in $H = \bigcap_{i = 1}^{s} \ker(\rho_{g_i})$.

In the following sections, we review the state of the art focusing on specific non-abelian groups, some related to problems of interest in cryptographic applications, others useful for understanding the current state of the art and the techniques developed for tackling the HSP in the non-abelian setting.
We also provide some insight on general solving approaches and an alternative to Fourier sampling.

\subsection{Dedekind and Hamiltonian groups}\label{ssec:ded}

These special classes of non-abelian groups are not of cryptographic relevance, but they are useful to identify successful approaches.

\begin{defi}
A group is called \emph{Dedekind} if all its subgroups are normal.
\end{defi}

\begin{theo} \cite{HRT00}
If $\rho_1, \ldots, \rho_s$ are $s = c \log |G|$ samples of the weak Fourier sampling for a normal hidden subgroup $H$, for $c>0$, then
$$
\text{Pr}\left[ H = \bigcap_{i = 1}^{s} \ker(\rho_i) \right] \geq 1 - \ex{- \frac{(c - 2)^2}{2 c} \log |G|}.
$$
\end{theo}

Thus, considering the intersection of the measured irreducible representations still gives information about $H$ if it is normal.
In particular, this result implies that the weak Fourier sampling can solve every instance of the HSP over Dedekind groups.

\begin{defi}
Non-abelian Dedekind groups are called \textit{Hamiltonian} groups and have the form $G = \mathbb{Z}_2^k \times A \times Q_8$ where $k$ is a non-negative integer, $A$ is an abelian group with all elements of odd order and $Q_8 = \{\pm 1, \pm \im, \pm \mathrm{j}, \pm \mathrm{k}\}$ is the quaternion group (whose only non-trivial subgroups are $\langle -1 \rangle$, $\langle \im \rangle$, $\langle \mathrm{j} \rangle$, and $\langle \mathrm{k} \rangle$).
\end{defi}

It is noteworthy that, by executing the weak Fourier sampling on multiple subgroups, the solution of the HSP can be generalized from the Dedekind case to the Hamiltonian case, see \cite{HRT03}.
In particular, it is shown in \cite{GSV01} that, for every subgroup $H$, this algorithm allows to find with high probability the largest subgroup of $H$ which is normal in $G$.
In general, a solving algorithm exists when the intersection of normalizers of $G$ is large enough, see \cite{Gav04}.

\subsection{Dihedral groups and the Shortest Vector Problem}
\label{ssec:dihedral}

In this section, we examine the HSP for the dihedral group, which is connected to the Shortest Vector Problem (SVP), on which the security of lattice-based cryptography is based.

\begin{defi}
The \emph{dihedral HSP} (DHSP) is the hidden subgroup problem for the dihedral group $D_N\cong \mathbb Z_N\rtimes_\varphi \mathbb Z_2$.
\end{defi}

As a first step, it is useful to understand what are the possible hidden subgroups of $D_N$.

\begin{prop}\label{prop:subsDN}
The subgroups of $D_N$ are:
\begin{itemize}
\item $H_d = d \mathbb{Z}_N \times \{0\} = \{(k d, 0) \,|\, 0 \leq k  < N / d\}$ for a divisor $d$ of $N$;
\item $H_{d, r}  = H_d \cup \Big( (r, 1)+ H_d\Big) = \{(k d, 0), (r + k d, 1) \,|\, 0 \leq k < N / d\} \cong D_{N/d}$ for a divisor $d$ of $N$ and $0 \leq r < d$.
\end{itemize}
\end{prop}

This is even more interesting considering the following result.

\begin{prop}\cite{EH98}
The DHSP reduces to efficiently find $r$ for the hidden subgroup $H_{N, r} = \langle (r, 1) \rangle = \{(0, 0), (r, 1)\}$.
\end{prop}

The same paper also proposed an algorithm for solving the DHSP in this case.
Firstly, the values $r = 0$ and $r = N / 2$ are checked: since the hidden subgroup $H_{N, r}$ has only two elements, it is sufficient to compare $f(0, 0)$ with $f(0, 1)$ and $f(N / 2, 1)$.
If none of these comparisons is satisfied, the following routine, which is practically the one for the abelian cases, is run:
\begin{algo}\label{algo:DHSP} \hfill\vspace{-0.3cm}
\begin{enumerate}
\item Given the hiding function $f\colon D_N \longrightarrow S$, with $D_N\cong \mathbb Z_N\rtimes_\varphi \mathbb Z_2$, initialize in the state $\ket{0}$ two group registers representing $\mathbb{Z}_N$ and $\mathbb{Z}_2$, and one (or more) auxiliary register representing $S$.
\item Create a uniform superposition on the group registers by applying the operator $\mathcal H^{\mathbb{Z}_N} \otimes \mathcal H \otimes \mathcal I^{S}$.
\item Compute the quantum implementation $\mathcal U_f$ of $f$ with the group registers as input and the auxiliary register containing the output, thus obtaining the state
$$
\frac{1}{\sqrt{2N}} \sum_{a = 0}^{N - 1} \sum_{b = 0}^1 \ket{a} \ket{b} \ket{f(a, b)}.
$$
\item Measure the auxiliary register to a value $f(\bar{a}, \bar{b})$, so that the group registers collapse to
$$
\frac{1}{\sqrt{2}} \big(\ket{\bar{a}} \ket{\bar{b}} + \ket{\bar{a} + (- 1)^{\bar{b}} r} \ket{\bar{b} + 1}\big) 
= \begin{cases}
\frac{1}{\sqrt{2}} \big(\ket{\bar{a}} \ket{0} + \ket{\bar{a} + r} \ket{1}\big), \\
\frac{1}{\sqrt{2}} \big(\ket{\bar{a}} \ket{1} + \ket{\bar{a} - r} \ket{0}\big).
\end{cases}
$$ 
Since $(\bar{a}, \bar{b})$ is chosen uniformly at random, this state is equivalent to taking $\bar{a} \in \mathbb{Z}_N$ uniformly at random and considering
$$
\frac{1}{\sqrt{2}} \big(\ket{\bar{a}} \ket{0} + \ket{\bar{a} + r} \ket{1}\big).
$$
\item Apply the QFT $\mathcal F_N$ to the first group register and obtain
$$
\frac{1}{\sqrt{2N}} \sum_{x = 0}^{N - 1} \omega^{\bar{a} x} \ket{x} \ket{0} + \omega^{(\bar{a} + r) x} \ket{x} \ket{1},
$$
where $\omega = \ex{\frac{2 \pi \im}{N}}$.
\item Measure the first group register, thus obtaining a random value $\bar{x} \in \mathbb{Z}_N$ and collapsing the second group register into the state
$$
\frac{\omega^{\bar{a} \bar{x}} }{\sqrt{2}} \big(\ket{0} + \omega^{r \bar{x}} \ket{1}\big).
$$
\item Apply $\mathcal H$ to the remaining group register, obtaining
\begin{align*}
&\frac{\omega^{\bar{a} \bar{x}}}{2}  \big(\ket{0} + \ket{1} + \omega^{r \bar{x}} (\ket{0} - \ket{1})\big) = \\
&\qquad = \frac{\omega^{\bar{a} \bar{x}}}{2}  \big((1 + \omega^{r \bar{x}}) \ket{0} + (1 - \omega^{r \bar{x}}) \ket{1}\big) \\
&\qquad = \omega^{\bar{a} \bar{x}} \omega^{r \bar{x}/2} \left( \frac{\omega^{-r \bar{x}/2} + \omega^{r \bar{x}/2}}{2} \ket{0} + \frac{\omega^{-r \bar{x}/2} - \omega^{r \bar{x}/2}}{2} \ket{1} \right) \\
&\qquad = \ex{\frac{\pi\im}{N} (2 \bar{a} + r) \bar{x}} \left( \cos\left(\frac{\pi}{N} r \bar{x}\right) \ket{0} - \im \sin\left(\frac{\pi}{N} r \bar{x}\right) \ket{1} \right).
\end{align*}
\item Measure the register to a bit $\bar{y} \in \mathbb{Z}_2$.
Considering that in Step~6 the output was taken randomly, the probability of picking the obtained result is given by 
$$
\text{Pr}[\bar{x}, \bar{y}]
=\text{Pr}[\bar{x}]\text{Pr}[\bar{y}]
= \frac{1}{N}\text{Pr}[\bar{y}]
= \begin{cases}
\frac{1}{N} \cos^2 \left( \frac{\pi r \bar{x}}{N} \right) & \text{if } \bar{y} = 0, \\
\frac{1}{N} \sin^2 \left( \frac{\pi r \bar{x}}{N} \right) & \text{if } \bar{y} = 1.
\end{cases}
$$
\end{enumerate}
\end{algo}
Then it is sufficient to have $O(\log N)$ samples $\bar{x}_1, \bar{x}_2, \ldots$ to find $r$ by searching the minimum of a function $x \longmapsto g(x, \bar{x}_1, \bar{x}_2, \ldots)$ in the domain $\{1, 2, \ldots, \lfloor N / 2 \rfloor\}$.
The whole algorithm allows to solve the DHSP using $O(\log N)$ calls to $f$.
However, to find the minimum of the function $g$, the tests required are $O(N)$.

Since finding $r$ is crucial to solve the DHSP, the state obtained in Step~4 is very important and it is worth focusing on it separately.

\begin{defi}
The \emph{Dihedral Coset Problem} (DCP) consists in finding $r \in \mathbb{Z}_N$ given a black box that outputs states of the form $\frac{1}{\sqrt{2}} \big(\ket{x} \ket{0} + \ket{x + r} \ket{1}\big)$ for $x \in \mathbb{Z}_N$.
\end{defi}

\begin{prop}
The DHSP and the DCP are computationally equivalent.
\end{prop}
\begin{proof}
Knowing a solution to the DCP means that we can find $r$ after Step~4 in \cref{algo:DHSP} above.
Conversely, given an algorithm consisting of Steps~4--8 of \cref{algo:DHSP}, with as input the output of the DCP black box that replaces the coset samples, the final output is a dihedral subgroup of the form $\{(0,  0), (r, 1)\}$, from which it is easy to obtain $r$ and thus solve the DCP.
The efficiency of these reduction follows directly from the polynomial-time complexity of the subroutines.
\end{proof}

It is noteworthy that there is a subexponential-time algorithm to solve the DCP, and thus the DHSP, due to Kuperberg, see \cite{Kup05}.

Focusing on the state of the second group register in Step~6 of \cref{algo:DHSP} that, ignoring the phase factor, is given by
$$
\ket{\psi_{\bar{x}}} = \frac{1}{\sqrt{2}} \big(\ket{0} + \omega^{r \bar{x}} \ket{1}\big),
$$
the idea behind the Kuperberg's algorithm can be easily described when $N = 2^n$.
In this case, he proved that, starting from states $\ket{\psi_{\bar{x}}}$ for different states $\bar{x} \in \mathbb{Z}_{2^n}$, it is possible to create the state 
$$
\ket{\psi_{2^{n - 1}}} 
= \frac{1}{\sqrt{2}} \big(\ket{0} + \ex{\pi \im r} \ket{1}\big)
= \frac{1}{\sqrt{2}} \big(\ket{0} + (- 1)^r \ket{1}\big),
$$
with high probability.
Applying the Hadamard gate to this register results in
$$
\mathcal H(\ket{\psi_{2^{n - 1}}}) 
= \frac{1 + (- 1)^r}{2} \ket{0} + \frac{1 - (- 1)^r}{2} \ket{1},
$$
so that, after measuring, if 0 is observed, then $r$ is even; otherwise, it is odd.

In particular, Kuperberg noted that, starting from $f\colon D_N \longrightarrow S$ hiding $H = \langle(r,1)\rangle$:
\begin{itemize}
\item if $r$ is even, then the function $f'\colon D_{N / 2} \longrightarrow S$ given by $f'(a, b) = f(2 a, b)$ hides $H' = \langle (r / 2, 1) \rangle$;
\item if $r$ is odd, then $f''\colon D_{N / 2} \longrightarrow S$ given by $f''(a, b) = f(2 a + 1, b)$ hides $H'' = \langle ((r - 1) / 2, 1) \rangle$.
\end{itemize} 
Thus, by iterating the process, all bits of $r$ can be found in $\exp(O(\sqrt{n}))$ steps with $\exp(O(\sqrt{n}))$ required quantum space.

The work of Kuperberg was the first of many improvements to the solution of DCP/DHSP.
Right after Kuperberg's preprint, Regev, in the unpublished paper \cite{Reg04}, reduced the quantum space requirement to polynomial while increasing the running time to $\exp\big(O\big(\sqrt{n \log n}\big)\big)$, i.e., subexponential.
Later, Kuperberg found another improvement that requires again $\exp(O(\sqrt{n}))$ steps but only $O(n)$ quantum space if sided by $\exp(O(\sqrt{n}))$ classical space; see \cite{Kup13}.

The resolution of the DHSP by coset sampling, i.e., of the DCP, collected all those attentions because of its direct relation to lattice-based cryptography.
The first step to see this relation is to consider the following generalization of the DCP.
\begin{defi}
The \emph{two-point problem} consists of finding the array $r \in \mathbb{Z}_N^\ell$ given a black box that outputs states of the form $\frac{1}{\sqrt{2}} (\ket{x} \ket{0} + \ket{x + r} \ket{1})$ for $x \in \mathbb{Z}_N^\ell$.
\end{defi}

Clearly, the DCP is a two-point problem with arrays of length $\ell = 1$, but Regev proved that it is possible to reduce the two-point problem to the DCP, thus obtaining an equivalence.

\begin{prop} \cite{Reg04}
The two-point problem reduces to efficiently solve the DCP.
\end{prop}

This result is very important for lattice-based cryptography, one of the most promising post-quantum solutions, because its security is generally based on the following problem.

\begin{defi}
A \emph{lattice} $\Lambda$ is a discrete subgroup of a Euclidean vector space of dimension $n$ (usually $\mathbb R^n$).
Its \emph{rank} rk$(\Lambda)$ is the maximum number of linearly independent vectors in $\Lambda$ and the lattice is called \emph{full rank} if rk$(\Lambda) = n$.
\end{defi}

\begin{defi}
Given a lattice $\Lambda \subset \mathbb{R}^n$ the \emph{Shortest Vector Problem} (SVP) consists in finding
$$
{w} \in \Lambda : \Vert {w} \Vert = \min_{{v} \in \Lambda - \{{0}\}} \Vert {v} \Vert,
$$
where $\| \cdot \|$ is the Euclidean norm.
\end{defi}

\begin{defi}
 Given an approximation function $f$ (e.g., $O$,  $\text{poly}$, $\text{exp}$) of the dimension $n$, a non-zero vector $w$ in a lattice $\Lambda\subset\mathbb R^n$ is \emph{$f(n)$-unique} if $\|v\|>f(n)\|w\|$ for each $v\in \Lambda-\mathbb Zw$.
\end{defi}

It is well known that the LLL algorithm by \cite{LLL82} allows to find the shortest vector in lattices with a $O(2^n)$-unique shortest vector.
Thus, in this case, the SVP is solvable in polynomial time.
It is then conjectured that the $O(1)$-unique SVP is the hardest case, so that the most interesting instances are those in the middle, like the $\text{poly}(n)$-unique SVP.

This has firstly been addressed by Regev in an unpublished work \cite{Reg04}.
He used a solution to the two-point problem to solve the $\Omega(\sqrt{n})$-unique SVP, i.e., when $\sqrt{n}$ is an asymptotic lower bound.
The method is briefly described in the following algorithm.
\begin{algo}\hfill\vspace{-0.3cm}
\begin{enumerate}
\item Create a superposition over the points in $\mathbb{Z}^n$.
Clearly, this is impossible since $\mathbb{Z}$ contains infinite values.
For the purpose of the algorithm, it is sufficient to take $M$ large enough, e.g., $M = \exp(O(n))$, and create the superposition
$$
\ket{\zeta} = \frac{1}{\sqrt{M^n}} \sum_{{x} \in \mathbb{Z}_M^n} \ket{{x}}.
$$
\item The lattice superposition is then obtained by applying $x \longmapsto \sum_{i = 1}^n x_i {b}_i$, where $\{{b}_1, \ldots, {b}_n\}$ is a basis of the lattice, possibly obtained with LLL and therefore reduced.
\item Now, a partition of $\mathbb{R}^n$ into cubes with exactly two points in each cell is required.
The quantum implementation of this operation is a labeling function $\pi(x, \text{par})$, depending on some parameters $\text{par}$, whose result is stored in an auxiliary register.
\item After measuring this register, i.e., fixing a specific cube randomly, the lattice register collapses to a superposition of the two points inside this cube.
This is an instance of two-point problem with the shortest vector of the lattice as solution.
\end{enumerate}
\end{algo}

Actually, the partition requested at Step~3 is efficiently obtainable only if the shortest vector is known.
Regev avoided this requirement by exploiting the $\Omega(\sqrt{n})$-uniqueness to construct a partition that contains at most two points in each cell.
This approximation of the optimal partition is good enough for Step~4 to succeed.

In recent years, new strategies have been introduced presenting different trade-offs between efficiency and space.
Some of the main results are briefly described below.

In 2014, Li, Bao and Fu \cite{LBF14} presented a new algorithm for the dihedral HSP based on the algorithm Shortest Vector (SV) introduced in \cite{LO85} by Lagarias and Odlyzko in 1985 (which exploits the LLL algorithm).
This algorithm takes $O(n)$ quantum space and $O(n^2)$ classical space and time polynomial for $n \leq 6400$.

In 2016, R\"otteler exploited a solution of the so-called hidden shift problem (see \cref{ssec:code}) to solve the HSP over $D_N$, with $N = (q^{d + 1} - 1) / (q - 1)$ for $q$ constant and $d$ corresponding to the input size of the problem.
His algorithm solves the HSP in $O(\log N)$ queries to the hiding function, $O(\text{poly}(\log N))$ quantum time, $O(\log N)$ quantum space and trivial classical post-processing, see \cite{Roe16}.
Another proof that the query complexity of the dihedral HSP is polynomial can be found in \cite{Bel19}.

In 2022, Fujita proposed a heuristic approach to solve the HSP over the dihedral group $D_N$ with $N = 2^n$ \cite{Fuj20}.

For interesting surveys devoted to the DHSP, see \cite{KG06, CS24}.

\subsection{Symmetric groups and the graph isomorphism problem} \label{ssec:sym}

Another interesting instance of non-abelian HSP is when the group is symmetric, namely the set $S_n$ of all the permutations that can be performed on $n$ elements with the composition operation. This has received a lot of attention in research because there is a reduction from the graph isomorphism problem, a famous unsolved computational problem, to the symmetric HSP \cite{CD10}.

However, despite some useful results were obtained in the efficient implementation of the Fourier transform for $S_n$, see \cite{Bea97}, it is proven that both weak Fourier sampling \cite{GSV01,HRT00} and strong Fourier sampling \cite{MRS05, KS05, DMR11} are incapable of solving the symmetric HSP in polynomial time.

As a final remark, we notice that even though the Cayley's theorem states that each group of cardinality $n$ can be homomorphically included in $S_n$, and, hence, the general HSP can be formally reduced to the symmetric HSP, we are not aware of any practical reduction that does not presuppose knowledge of the hidden subgroup.

\subsection{Semidirect products} \label{ssec:semi}
Because of the interpretation of the dihedral group as a semidirect product, groups with the same construction have received most of the attention of the community.

One of the first results is in an unpublished work \cite{RB98} from 1998, where Rötteler and Beth used Fourier sampling over normal subgroups to solve the HSP over the group $(\mathbb{Z}_2^n \times \mathbb{Z}_2^n) \rtimes_\varphi \mathbb{Z}_2$, where the elements of $\mathbb{Z}_2$ work as permutations over the positions of $\mathbb{Z}_2^n \times \mathbb{Z}_2^n$
(this can also be referred to as wreath product $\mathbb{Z}_2^n\wr\mathbb{Z}_2$).
More specifically, we have
\begin{align*}
\varphi(x)\colon \mathbb{Z}_2^n \times \mathbb{Z}_2^n &\longrightarrow \mathbb{Z}_2^n \times \mathbb{Z}_2^n ,\\
(n_1, n_2) &\longmapsto ((x+1) n_1 + x n_2, x n_1 + (x+1) n_2),
\end{align*}
for $x \in \mathbb Z_2$.
The resulting algorithm has a quantum routine that requires polynomial time, but exponential classical post-processing.

Another interesting result was obtained by Grigni, Schulman, Vazirani and Vazirani in 2004 \cite{GSV01} (after a publication in the proceedings of STOC 2001), where they considered the cyclic subgroups $\langle 3 \rangle$, $\langle 2^n \rangle$ of $\mathbb Z_{3 \cdot 2^n}$ and studied the inner semidirect product $\langle 2^n \rangle \rtimes \langle 3 \rangle$ as an example of a more general case.
In particular, they focus on using weak Fourier sampling over subgroups for the HSP over groups $G$ with $|G| / |M_G| \in \exp\big(O\big(\sqrt{\log \log |G|}\big)\big)$, where $M_G = \cap_{H \leq G} N(H)$ is the intersection of normalizers in $G$.

Other examples are group extensions $G$ of abelian groups $A$ by $B$ for $A \lhd G$ and $G / A \cong B$.
This work was generalized by Gavinsky in 2004 \cite{Gav04} to groups $G$ with $|G| / |M_G| \in O(\log |G|)$ and to the cases in which the hidden subgroup $H$ satisfies $|G| / |H M_G| \in O(\log |G|)$.

New achievements with the direct products were obtained by Friedl, Ivanyos, Magniez, Santha, and Sen in 2002 (then published in 2014) \cite{FIMSS14}.
They solved the HSP over groups $\mathbb{Z}_p^k \rtimes_\varphi \mathbb{Z}_{2}$ with $p$ prime and $\varphi(x)\colon n \mapsto (-1)^x n$, in quantum time $\text{poly}(k)$ by using the hidden shift problem (which they call hidden translation) while adopting the classical algorithm.

In their 2005 preprint, then published in 2007 \cite{MRRS07}, Moore, Rockmore, Russell and Schulman started a new line of research.
Considering the affine group  $\text{Aff}(1, p)$, its subgroup of translations is isomorphic to $\mathbb Z_p$ and its multiplicative subgroup is isomorphic to $\mathbb Z_p^\times \cong \mathbb Z_{p-1}$, so that $\text{Aff}(1, p) \cong \mathbb Z_p \rtimes \mathbb Z_{p-1}$.
Starting from this, we can similarly consider the so-called $q$-hedral groups $\mathbb Z_p \rtimes \mathbb Z_q$ with $q$ dividing $p-1$ and $p / q = \text{poly}(\log p)$; the authors used strong Fourier sampling for the HSP over these groups.
In addition, they generalize the result on group extensions by exploiting the following: given $A \lhd G$ such that $G / A \cong B$ with $A$ of size $\text{poly}(\log |B|)$ and assuming that there is an efficient HSP algorithm for $B$, then there is also an efficient HSP algorithm for $G$.
Unfortunately, this construction cannot be iterated more than a constant number of times, whereas in most cases it is required to repeat it a polynomial number of times.

In 2005, Bacon, Childs and van Dam \cite{BCvD05} exploited the Pretty Good Measurement (PGM) they introduced in \cite{BCvD06} to give a general proposal for $A \rtimes_\varphi \mathbb{Z}_p$ where $A$ is an abelian group, $p$ is a prime, and $\varphi$ is a general homomorphism from $\mathbb Z_p$ into $\text{Aut}(A)$.
Firstly, they apply a reduction similar to that of the dihedral HSP.
Thus, by solving the abelian HSP over $A$, they find $H_1 = H \cap (A \times \{0\})$ and, despite $H_1$ may not be normal, they find a way to use quotient reduction.
This reduces the problem to the HSP over $A_2 \rtimes_\varphi \mathbb{Z}_p$ with $A_2 \leq A$ and hidden subgroup trivial or of size $p$ generated by $(a, 1)$ for some $a \in A_2$.
Next, the PGM from \cite{BCvD06} is applied to a product of $k$ coset states, and the problem is reduced to the so-called \textit{matrix sum problem}.
They showed that the HSP over $\mathbb{Z}_{N} \rtimes_\varphi \mathbb{Z}_{p}$, with $N$ arbitrary, $p$ prime, and $N / p = \text{poly}(\log N)$, can be solved in time $\text{poly}(\log N p)$, thus generalizing the result in \cite{MRRS07}.
Moreover, they proved that the HSP over $\mathbb{Z}_p^k \rtimes_\varphi \mathbb{Z}_{p}$, with $k$ fixed and $p$ prime, can be solved in time $\text{poly}(\log p)$. 
In conclusion, the PGM proved to be optimal for distinguishing conjugate hidden subgroups \cite{MR07}.

Starting from these results, in 2009 Gonçalves, Portugal and Cosme \cite{GPC09} obtained a polynomial-time quantum algorithm for the HSP over $\mathbb{Z}_{p} \rtimes_\varphi \mathbb{Z}_{q^s}$ for $p$ and $q$ odd primes with $q \mid p - 1$ and $p / q = \text{poly}(\log p)$.
Without these constraints, this algorithm is subexponential.
The homomorphism $\varphi$ that defines the semidirect product is completely determined by specifying $\phi(1)(1) = \alpha \in \mathbb Z_p^\times$, since $\text{Aut}(\mathbb Z_p)$ is isomorphic to $\mathbb Z_p^\times$.
In their work, the authors focus on a specific value for $\alpha$.
They further generalize their result in an unpublished work \cite{GP11} from 2011, where they consider $\mathbb{Z}_{p^r} \rtimes_\varphi \mathbb{Z}_{q^s}$ for $p \neq q$ odd primes with $p^r / q = \text{poly}(\log p^r)$.

Finally, the highest level of generalization was achieved by van Dam and Dey \cite{vDD14} in 2014 on $\mathbb{Z}_{p^r} \rtimes_\varphi \mathbb{Z}_{q^s}$ for $p$ and $q$ odd primes with $p^r / q^{t - j} = \text{poly}(\log p^r)$, where $0 \leq j <t$ are parameters specific to the group.

Independently of the described chain of improvements, Krovi and Rötteler \cite{KR08} managed to solve the HSP over Weyl-Heisenberg groups $\mathbb{Z}_p^{k + 1} \rtimes_\varphi \mathbb{Z}_p^k$ through weak Fourier sampling.

\subsection{Semidirect products and code equivalence}
\label{ssec:code}

We show here another reduction to the HSP of a well--know problem of interest in cryptography: the \emph{code equivalence problem}. It is strictly connected with the topic of semidirect products.
We start with some basic definitions.

\begin{defi}
Let $1 \le k \le n$ be two integers. A \emph{linear code} $\mathcal{C}$ of \emph{length} $n$ and \emph{dimension} $k$ over the finite field $\mathbb{F}_q$, with $q$ a prime power, also called a $[n,k]_q$ code, is a $k$-dimensional linear subspace of the vector space $\mathbb{F}_q^n$.
A matrix $M \in \mathbb{F}_q^{k \times n}$ is a \emph{generator matrix} of $\mathcal{C}$ if its rows form a basis for $\mathcal{C}$, i.e., if
\begin{equation*}
\mathcal{C} = \{ xM \, | \, x \in \mathbb{F}_q^k \}.
\end{equation*}
\end{defi}

In the post-quantum scenario, codes are used to build key-encapsulation mechanisms and signatures schemes.
Beside the well-known \emph{syndrome decoding problem}, some schemes, such as \cite{BMP20}, exploit the \emph{code equivalence problem}.
Here, we define a version of it which turns out to be a sub-instance of the main problem.

\begin{defi}
    Given two $[n,k]_q$ codes generated by matrices $M$ and $M^\prime$ respectively and such that 
    \begin{equation*}
        M^\prime = S M P
    \end{equation*}
    for some invertible matrix $S \in \gl_k(\mathbb{F}_q)$ and $n \times n$ permutation matrix $P \in \mathcal{S}_n \subseteq \gl_n(\mathbb F_q)$,
    the \emph{Permutation code Equivalence Problem} (PEP) consists of finding such a pair $(S,P) \in \gl_k(\mathbb{F}_q) \times \mathcal{S}_n$, where $\mathcal S_n$ is the subgroup of permutation matrices in $\gl_n(\mathbb F_q)$.
\end{defi}

For further details on coding theory and its connection with cryptography, we refer to \cite{WGR22}.

The reduction from PEP to HSP goes through an intermediate step given by the following problem.

\begin{defi} \label{def:hiddenshift}
    Let $G$ be a finite group and $X$ a finite set.
    Given two functions $f_0, f_1\colon G \to X$ for which there is an element $s \in G$ such that
    \begin{equation*}
        f_1(g) = f_0(sg), \qquad \text{for every } g \in G,
    \end{equation*}
    the group element $s$ is called (\emph{left}) \emph{shift} from $f_0$ to $f_1$.
    The \emph{hidden shift problem} consists of finding the shift element knowing the two functions.
\end{defi}

\begin{prop} \cite{DMR11}
\label{pep:hiddenshift:reduction}
    The permutation code equivalence polynomial-time reduces to a hidden shift problem over the group $\gl_k(\mathbb{F}_q) \times \mathcal{S}_n$.
\end{prop}

\begin{proof}
    Let $G$ be the group
    \begin{equation} \label{codeequiv:HS:group}
        G = \gl_k(\mathbb{F}_q) \times \mathcal{S}_n,
    \end{equation}
    with component-wise group operation
    \[
    (S_1, P_1) \cdot (S_2, P_2) = (S_1 S_2, P_1 P_2).
    \]
    Fixing the $k \times n$ matrices $M$ and $M^\prime$, we define the functions
    \begin{equation} \label{codeequiv:HS:functions}
    \begin{split}
        f_0 (S,P) & := S^{-1} M P, \\
        f_1 (S,P) & := S^{-1} M^\prime P.
    \end{split}
    \end{equation}
    It is clear that
    \[ \bar S M \bar P = M^\prime \text{ for some $\bar  S$ and $\bar P$} \quad \text{if and only if} \quad (\bar S^{-1},\bar P) \text{ is a shift from $f_0$ to $f_1$}.\]
    Hence, solving in polynomial time the hidden shift problem over $G$, given functions $f_0$ and $f_1$ above defined, implies solving efficiently the permutation code equivalence for the codes generated by $M$ and $M^\prime$.
\end{proof}

The second step of the reduction is not straightforward: in fact, not every hidden shift instance can be recast as a HSP one. In order to get the result, we need an additional assumption on the shifted function $f_0$, namely, the fact that it is hiding itself some subgroup.

\begin{prop} \cite{DMR11}
    If the function $f_0\colon G \to X$ in \Cref{def:hiddenshift} hides a subgroup $H_0 < G$, then the hidden shift problem polynomial-time reduces to an instance of the hidden subgroup problem over the group $G^2 \rtimes \mathbb{Z}_2$.
\end{prop}

\begin{proof}
    Let $f_0$ and $f_1$ be two functions from $G$ to the finite set $X$ for which there exists a left shift from the former to the latter, and let $G^\prime$ be the group 
    \begin{equation} \label{codeequiv:HSP:group}
        G^\prime = G^2 \rtimes \mathbb{Z}_2.
    \end{equation}
    For every element $((g,h),b) \in G^2 \rtimes \mathbb{Z}_2$, we define the function $f\colon G^\prime \to X \times X$ as
    \begin{equation} \label{codeequiv:HS:function}
        f((g,h), b) = 
        \begin{cases}
            (f_0(g),f_1(h)) & \text{if } b=0, \\
            (f_1(h),f_0(g)) & \text{if } b=1.
        \end{cases}
    \end{equation}
    We show that solving the hidden shift problem is equivalent to recovering the subgroup $H^\prime < G^\prime$ hidden by the function $f$, exploiting the assumption that $f_0$ hides the subgroup $H_0$. \par
    One can check that, under our assumption, $f$ hides the subgroup
    \begin{equation*}
        H^\prime = \bigl( (H_0, s^{-1}H_0s), 0 \bigr)\cup \bigl( (H_0s, s^{-1}H_0),1\bigr),
    \end{equation*}
    where $s$ is a shift from $f_0$ to $f_1$ and $H_0$ is the subgroup of $G$ hidden by $f_0$. \par
    Once one has found the hidden subgroup for $f$, it is easy to recover from that a solution of the hidden shift. In order to do so, it is important to notice that the set of all the shifts from $f_0$ to $f_1$ is given by $H_0s$, with $s$ a particular shift. Then, the following procedure can be exploited:
    \begin{enumerate}
        \item sample from $H^\prime$ an element $((g_1,g_2),b)$;
        \item if $b=0$, discard it and go back to Step 1;
        \item if $b=1$, then $g_1 \in H_0s$, hence it is a shift.
    \end{enumerate}
    Therefore, solving in polynomial time the hidden subgroup problem over $G^\prime$ given by the function $f$ allows us to define an efficient procedure to solve the hidden shift problem.
\end{proof}

We have not yet reduced PEP to HSP: in order to do so, we have to check that the hypothesis required by the previous proposition holds for the reduction shown in \Cref{pep:hiddenshift:reduction}. The following result claims exactly this.

\begin{prop} \cite{DMR11}
    The function $f_0$ defined in \Cref{codeequiv:HS:functions} hides the subgroup
    \begin{equation} \label{codeequiv:subgp}
        H_0 = \{(S,P)\in \gl_k(\mathbb{F}_q) \times \mathcal{S}_n \, | \, S^{-1}MP=M \} < \gl_k(\mathbb{F}_q) \times \mathcal{S}_n. 
    \end{equation}
\end{prop}
\begin{cor}
    The permutation code equivalence problem between two $[n,k]_q$ codes efficiently reduces to a hidden subgroup problem over $(\gl_k(\mathbb{F}_q) \times \mathcal{S}_n)^2 \rtimes \mathbb{Z}_2$.
\end{cor}

Since the permutation code equivalence can be recast as an instance of the HSP on a semidirect product of groups, the techniques discussed in \Cref{ssec:semi} apply to find its solution.

\subsection{Fourier sampling on other groups}
\label{ssec:other}

Despite most of the work focused on the HSP over particular instances of semidirect product, there were sporadic cases dealing with other group structures.

For instance, in 2007, Ivanyos, Sanselme, and Santha \cite{ISS07} solved the HSP over $p$-groups $G$ (groups with order a power of $p$) that are \emph{extra-special}, i.e., with center $Z(G) := \{z \in G \,|\,  z g = g z, \,  g \in G \}$ cyclic of order $p$ and $G / Z(G)$ elementary abelian (we recall that an elementary abelian group is a group such that its non-trivial elements have same order).
Later, in \cite{ISS10}, they extended the result to nil-2 groups (groups with $[[G, G], G] = \{1\}$, where $[G, G] := \langle\{g^{- 1} h^{- 1} g h \,|\, g, h \in G\}\rangle$ is the commutator of $G$) by exploiting a generalization of the hiding oracle $f$ introduced in the 2002 preprint of \cite{FIMSS14}.
In particular, they define a \emph{hiding procedure} for a subgroup $H$ of $G$ that allows to create quantum states which are orthogonal depending on the coset of $H$.
Since the abelian HSP algorithm still works when the oracle is replaced by such a procedure, they reduced the HSP over nil-2 groups to those having exponent $p$ (i.e., $p = \text{lcm}\{\text{ord}(g) \,|\, g \in G\}$) showing that a hidden subgroup $H$ can be found if there is an efficient hiding procedure for $H [G, G]$.

In 2010, the extension introduced in \cite{MRRS07} was used by Denney, Moore, and Russell \cite{DMR10} to solve the HSP over the special linear group $\text{SL}(2,p^r)$, the projective special linear group $\text{PSL}(2, p^r) := \text{SL}(2,p^r) / \{\pm \mathrm{Id}\}$, and the projective general linear group $\text{PGL}(2,p^r) := \text{GL}(2,p^r) / \{\text{scalar matrices}\}$ when the hidden subgroups are stabilizers $G_s := \{g \in G \,|\, g s = s\}$, for $s \in S$, where $S$ is the codomain of the hiding function $f$.
In particular, their idea was to restrict themselves to the standard Borel subgroup $B$ of the upper triangular matrices, which is isomorphic to $\text{Aff}(1, p^r)$, and to apply a strong Fourier sampling on this group determining $B \cap G_s$ and deducing~$s$.

Further developments on Borel subgroups have been obtained in \cite{Iva12} on the general linear group $\text{GL}(n, q)$ for particular values of $q$ not much smaller than $n$, while \cite{DIK14} focuses on hidden parabolic subgroups finding a solving algorithm polynomial in $\log q$ and $n$.

It is important to notice that all the described results work for groups which have something in common with abelian groups.
For example, the procedure for the Dedekind HSP works also for groups with many normal subgroups, as well as for semidirect products of abelian groups $G = A \rtimes B$ with $A \cong A \times \{0\} \lhd G$ and $B \cong G / (A \times \{0\})$.

\subsection{General approaches for non-abelian groups}
\label{ssec:general}

As observed in the previous sections, the literature is full of different results on the solution of the HSP in specific types of finite groups.
Here we describe some theoretical attempts to find a general strategy for non-abelian groups, always based on Fourier sampling.

One of the first general results was obtained in 2004 by Ettinger, Høyer and Knill.
They proved that the HSP over a finite group $G$ has quantum query complexity $\text{poly}(\log |G|)$, i.e., $O(\log |G|)$ queries to the function oracle provide sufficient statistical information to solve the HSP.
Despite this result being exponentially better than the classical alternatives, the general algorithm they provided to manipulate this information requires exponential time.
See \cite{EHK04}.

From the algebraic point of view, the HSP over non-abelian groups can be studied following two approaches.
The first consists in exploiting simple groups, the second in finding a reduction method.

\begin{defi}
A group $G$ is \emph{simple} if its normal subgroups are only $\{1\}$ and~$G$.
\end{defi}

The theorem of classification of finite simple groups states that every finite simple group is either cyclic of prime order, alternating of degree at least 5, in one of the 16 families of groups of Lie type, or one of 26 sporadic groups.

\begin{defi}
The \emph{composition series} of a finite group $G$ is a sequence
\[
\{1\} = H_0 \lhd H_1 \lhd \ldots \lhd H_l = G
\]
where each $H_i$ is a maximal normal subgroup of $H_{i + 1}$.
The integer $l$ is called \emph{length} of the composition and has size $O(\log |G|)$.
Each $H_{i + 1} / H_i$ is a simple group called \emph{composition factor}.\\ 
If all the composition factors of a group are abelian, the group is called \emph{solvable}.
\end{defi}

\begin{prop}
\label{th:wan1} \cite{Wan10}
The HSP over a generic finite group $G$ can be efficiently solved if both the following conditions hold:
\begin{enumerate}
\item there is an efficient algorithm for solving the HSP over simple groups;
\item for each HSP with $f$ on a non-simple $G$ hiding $H$, given a proper normal subgroup $N$ of $G$, the following functions are both efficiently implementable:
\begin{itemize}
    \item $\bar{f}$ on $G/N$ hiding $\bar{H} = \{h N \,|\, h \in H\}$;
    \item $f|_N$ on $N$ hiding $H \cap N$.
\end{itemize}
\end{enumerate}
\end{prop}

\begin{rem}
When the group is solvable, only the second condition is required.
\end{rem}
About the second approach, finding reduction methods consists of passing through one of the main reduction types: subgroups or quotients.

\begin{defi} \label{def:wang}
Given an instance of HSP with $f$ over a generic finite group $G$ hiding $H$, it is possible to apply:
\begin{itemize}
\item \textit{Subgroup reduction}: if it is possible to find in polynomial time a subgroup $H \subseteq G_1 \subseteq G$ and an isomorphism $g\colon G_2 \overset{\sim}{\longrightarrow} G_1$ computable in polynomial time, then $f|_{G_1} \circ g$ hides the subgroup $g^{- 1}(H)$ of $G_2$.
Hence, solving this HSP returns $g^{- 1}(H) = \langle u_i \rangle_i$ so that $\langle g(u_i) \rangle_i = H$.
\item \textit{Quotient reduction}: if it is possible to find in polynomial time a subgroup $H_1$ normal in $G$ contained in $H$, with $H_1 = \langle u_i \rangle_i$ and an isomorphism $g\colon G_2 \overset{\sim}{\longrightarrow} G / H_1$ computable in polynomial time, then $\bar f\circ g$ (where $\bar f$ is the map induced by $f$ on $G/H_1$) hides the subgroup $g^{- 1}(H / H_1)$ of $G/H_1$.
Hence, solving this HSP returns $g^{- 1}(H / H_1) = \langle v_j \rangle_j$ so that $\langle g(v_j) \rangle_j = H / H_1$ and $\langle u_i, g(v_j) \rangle_{i,j} \cong H$.
\end{itemize}
\end{defi}

The first reduction can be applied recursively to obtain the following result.

\begin{theo} \label{th:wan2} \cite{Wan10}
An instance of HSP with $f$ on $G$ hiding $H$ can be efficiently solved if there is a polynomial-time algorithm that finds a maximal subgroup $G_1$ containing $H$.
\end{theo}

The second reduction from \cref{def:wang} can lead to the following result.

\begin{theo} \label{th:wan3} \cite{Wan10}
An instance of HSP can be efficiently solved if, after recursively repeating the quotient reduction over the largest normal subgroup included in the hidden subgroup, the trivial group is the only normal subgroup in the reduced hidden subgroup.
\end{theo}

This is obtained because, in general, the weak Fourier sampling returns a normal subgroup of $G$ that is the largest contained in $H$, on which quotient reduction can be applied.
This procedure can be repeated in the quotient group until $\{1\}$ is the largest hidden subgroup of $H$ normal in $G$, that should happen after $O(\log |G|)$ reductions, see \cite{EHK04}.
However, when $H$ is not simple, other methods, such as subgroup reduction, are required, as seen in the case of the dihedral group.

With all these considerations, Wang, in his 2010 master's thesis \cite{Wan10}, proposed a possible algorithm for a general instance of HSP with $f$ on $G$ hiding $H$.
\begin{algo}\label{algo:wang}\hfill\vspace{-0.3cm}
\begin{enumerate}
\item Coset measurements allow to determine whether $H = G$ or not: supposing to repeat them, if $k$ measurements return the same value, then $H = G$ with probability at least $1 - 2^{- k}$.
In this case, return $G$.
Otherwise, if two different values are found, then $H \subsetneq G$.
\item If $G$ is a particular case with an efficient algorithm for solving the HSP, then execute it and return $H$.

\item If $G$ is simple, apply subgroup reduction over a maximal subgroup $G'$ containing $H$.
This can be repeated recursively until $H$ is found.
Return $H$.

\item If $G$ is not simple, follow \cref{th:wan3} by applying weak Fourier sampling and quotient reduction until possible.
If the last $H'$ found is trivial, then find and return $H$.

\item If the last $H'$ found is not trivial, use an alternative technique.
For example, search a normal subgroup of $G$ for \cref{th:wan1} or apply subgroup reduction as specified in \cref{th:wan2}.
If necessary and possible, repeat recursively to find and return $H$.
\end{enumerate}
\end{algo}

\begin{rem}
Clearly, this algorithm allows to solve the HSP when the group has a particular structure by means of the reductions in \cref{th:wan1}.
However, this is not universal because the last step can be very tricky.
\end{rem}

Applying the described approaches when $G$ is a Dedekind group allows to use the weak Fourier sampling to get a sequence of subgroups each containing the following one and $H$.
Since they are all normal, this results in $G = G_0 \rhd G_1 \rhd \ldots G_m \rhd H \rhd \{1\}$, that is a subsequence of the composition series of $G$.
Consequently, one can apply one step of quotient reduction, moving to $G/G_1$ that is simple, and then continue with repeated subgroup reductions as Step~3 in \cref{algo:wang}.

Now, we consider the dihedral group and the resulting reductions.
By \cref{prop:subsDN}, the possible composition series for $D_N$ are 
$$
\{(0, 0)\} = H_N \lhd \ldots \lhd
\begin{Bmatrix}
H_{2, 0} \cong D_{N / 2}, \\
H_1 \cong \mathbb{Z}_N, \\
H_{2, 1} \cong D_{N / 2}
\end{Bmatrix}
\lhd H_{1, 0} = D_N,
$$
if $N$ is even, and only the central line if $N$ is odd.
In particular, a HSP hiding $H_{d,r}$ can be solved by means of the following reductions:
\begin{itemize}
    \item $H_{d, r}$ contains the normal subgroup $H_d = H_1 \cap H_{d, r}$ so that, after a subgroup reduction to $H_1 \cong \mathbb{Z}_N$ (a cyclic group), the solution of the HSP for $H_1 \cap H_{d, r}$ allows to find $d$.
    \item Since the result $H_d$ is a normal subgroup of $D_N$ included in $H_{d, r}$, quotient reduction can be exploited.
    The obtained group is $D_N / H_d \cong D_N / \mathbb{Z}_{N / d} \cong D_d$, therefore the composition series is the same as in the previous diagram with $d$ replacing $N$, so that we can study directly the general case.
    \item The hidden subgroup is now $H_{N, r}$ with $r$ unknown, but quotient reduction is no more helpful.
    For subgroup reduction, the possible subgroups of $D_N$ are all $H_{a, b}$ with $a$ dividing $N$ and $b=r - qa$, so that $H_{a, b} \cong D_{N / a}$ and the resulting hidden subgroup is generated by $(q, 1)$.
    \item Repeating this subgroup reduction allows to determine $r$ in only $O(\log (2N))$ steps.
\end{itemize}
In conclusion, the dihedral HSP can be reduced to finding $b \equiv r \pmod{a}$, where $a$ divides $N$.

\subsection{Black-box approaches}
\label{ssec:black}

Since Fourier sampling is not always efficient, it is interesting to find methods for solving the HSP with alternative ways to sample from the relevant groups.
One promising proposal is the black-box groups approach on which we focus in this section.

\begin{defi}
A group $G$ of cardinality $N$ is a \emph{black-box group} if, given $m \geq 0$ and $n = O(\log N)$, all the following conditions hold:
\begin{enumerate}
\item $G$ has $m$ generators $g_1, \ldots, g_m$ that can be obtained efficiently;
\item multiplication and inversion can be performed efficiently through some oracles;
\item each element can be represented as a string of $n$ bits;
\item this binary representation is unique or it is possible to check element equality through an efficient oracle.
\end{enumerate}
Specifically, $n$ is the \emph{encoding length} and $m n$ is the \emph{input size}.

The black-box group has \emph{unique encoding} if in Condition 4 the first case is true.
\end{defi}

The first results on black-box groups were obtained in 2001, when Ivanyos, Magniez, and Santha \cite{IMS01} proposed quantum algorithms for solving the following:
\begin{itemize}
\item The normal HSP over solvable and permutation black-box groups in polynomial time (without any assumption on computability of non-commutative Fourier transforms).
\item The HSP over a black-box $G$ with unique encoding, with time polynomial in the input size plus the order of the commutator $[G,G]$.
In particular, this allows to solve the HSP over \emph{extra-special} $p$-groups in polynomial time in the input size plus $p$.
\item The HSP over black-box group $G$ with unique encoding and a normal elementary abelian 2-subgroup $N$ in polynomial time in the input size plus $|G / N|$.
This includes the particular cases where $G / N$ is cyclic.
\end{itemize}

In the work \cite{FIMSS14}, already mentioned among the results on Fourier sampling, Friedl, Ivanyos, Magniez, Santha, and Sen considered the following particular case of solvable black-box groups.

\begin{defi}
The \emph{derived series} of a group $G$ is defined inductively as
\[
\begin{cases}
G^{(0)} := G, \\
G^{(k + 1)} := [G^{(k)}, G^{(k)}],
\end{cases}
\]
so that $G=G^{(0)} \rhd G^{(1)} \rhd \ldots$.
An abelian group is \emph{smoothly abelian} if it can be expressed as the direct product of a subgroup of bounded exponent and a subgroup of polylogarithmic size in the order of the group.
A solvable group is \emph{smoothly solvable} if its derived series has bounded length and has smooth abelian factor groups.
\end{defi}

In particular, in \cite{FIMSS14} Friedl et al. introduced the \textit{orbit coset problem} which is a quantum generalization through the orbit superposition of hidden shift problem and HSP.
By finding an efficient algorithm for the orbit coset problem over solvable black-box groups with unique encoding that have smoothly solvable commutator subgroups, they also obtained an efficient solution of the HSP in that case.

Black-box approaches were also useful when dealing with semidirect products, hence we conclude this section with an historical excursus of these developments.

In their 2004 preprint, then published in 2007 \cite{ILG07}, Inui and Le Gall studied the HSP over $\mathbb{Z}_{N} \rtimes_\varphi \mathbb{Z}_q$ with $q$ prime.
They reduced the study to the cases $\mathbb{Z}_{p^r} \rtimes_\varphi \mathbb{Z}_q$ for $p$ prime and obtained a complete classification of these groups into five classes.
They proposed an efficient quantum algorithm for the class $\mathbb{Z}_{p^r} \rtimes_\varphi \mathbb{Z}_p$ for $p, r \neq 2$ and $\varphi(a_2)(a_1) = (a_2 p^{r - 1} + 1) a_1$, even if the black-box group is given without unique encoding.
They also generalized this algorithm for the case $\mathbb{Z}_{p^r}^k \rtimes_\varphi \mathbb{Z}_p$, where $\mathbb{Z}_p$ acts separately on each coordinate of $\mathbb{Z}_p^k$, if it is given as a black-box group with unique encoding.

Later, in 2006, Chi, Kim, and Lee \cite{CKL06} noted that the previous algorithm also allows to solve the HSP over $\mathbb{Z}_{2 p^r} \rtimes_\varphi \mathbb{Z}_p$ with $p$ an odd prime and found an extended version for $\mathbb{Z}_{N} \rtimes_\varphi \mathbb{Z}_p$ where $N = p_1^{r_1} \cdots p_n^{r_n}$ and $p$ does not divide any $p_i - 1$.
Their idea was to reduce the group to $\mathbb{Z}_{N/p_i^{r_i}} \times (\mathbb{Z}_{p_i^{r_i}} \rtimes_\varphi \mathbb{Z}_p)$ and then apply the algorithm of the 2004 preprint of \cite{ILG07}.

Using a similar approach, in 2007, Cosme and Portugal \cite{CP07} addressed the HSP over $\mathbb{Z}_{p^r} \rtimes_\varphi \mathbb{Z}_{p^2}$ with $p$ an odd prime and $r \ge 5$, and over the group $\mathbb{Z}_{N} \rtimes_\varphi \mathbb{Z}_{p^2}$ where $N = p_1^{r_1} \cdots p_n^{r_n}$ and $p$ does not divide any $p_i - 1$.

In 2017, Gonçalves, Fernandes, and Cosme \cite{GFC17} obtained another generalization for the HSP over $\mathbb{Z}_{N} \rtimes_\varphi \mathbb{Z}_{p^r}$ for $N = p_1^{r_1} \cdots p_n^{r_n}$ and when there is one $1 \le i \le n$ such that the odd prime $p$ divides $p_i - 1$, but does not divide $p_j - 1$ for each $j \ne i$.

Another recent result used black-box groups to solve the HSP in polynomial time over nilpotent groups with constant nilpotency class and whose order only have prime factors also bounded by a constant \cite{II24}.

It is noteworthy how the results obtained from these black-box methods allow to reduce the groups into abelian pieces through a normal series.
However, these methods are not suitable for groups which contain a non-cyclic simple group in their composition series such as $A_n$ (alternating) or $S_n$ (symmetric).

\section{Other directions}
\label{sec:concl}

Since we are interested in cryptographical applications, finite groups are the most relevant cases for this survey, but the scientific community has also worked on other interesting developments.
For instance, infinite groups have been addressed both in the abelian \cite{Tra23} and general case \cite{Kup25}, and also in relation to group-based cryptography \cite{HK18}. But there are also generalizations to Heisenberg groups \cite{RRS05,Bac08}, universal algebras \cite{MW20}, continuous vector spaces \cite{BDF20} or symmetries of a partition \cite{DISW13}.

Moreover, although computational and complexity aspects are outside our focus, the literature is full of works that obtain optimal implementations or reduce the required quantum queries of the algorithms for solving the HSP \cite{HKK08,FZ08,Sha14,SSV19,YL22a,YL22b,Nay22}.

\section*{Acknowledgment}
The work is partially funded by the Italian Ministry of University and Research in the framework of the Call for Proposals for scrolling of final rankings of the PRIN 2022 call - Protocol no.~2022RFAZCJ.
The second, third, and fourth authors are supported by the ``National Group for Algebraic and Geometric Structures, and their Applications" (GNSAGA - INdAM).

\printbibliography

\newpage

\section*{Appendix} \label{sec:app}

\paragraph{HSP on specific groups.}
For convenience of the readers, we collect in \Cref{tab:hsp} below the groups for which HSP has been studied, listing the methods used and the reference.

\renewcommand{\arraystretch}{1.5}
\begin{longtable}{p{3.0cm}|p{5.5cm}|p{2.5cm}|p{2.3cm}}
\caption{Collection of the main solving methods for the HSP by different group types in order of publication.
If not specified, the algorithm runs in polynomial time.}
\label{tab:hsp}\\

Group type & Notes & Method & Reference \\
\hline
\hline

Finite abelian
& $\mathbb{Z}_N$ 
& Fourier & \cite{Sho94} \\

& $\mathbb{Z}_N \times \mathbb{Z}_N$ 
& Fourier & \cite{Sho94} \\

& All 
& Fourier & \cite{Kit95} \\

& $\mathbb Z_{m^k}^n$  
& Fourier & \cite{II22} \\

& Solvable groups 
& Fourier & \cite{II22} \\

\hline

Dihedral $D_N$
\newline $(\mathbb{Z}_N \rtimes_\varphi \mathbb{Z}_2)$ 
& All 
& Fourier \newline (\textbf{exp}~classic) & \cite{EH98} \newline (arXiv 1998) \\

& All, SVP 
& Fourier \newline (\textbf{subexp}) & \cite{Reg04} \\

& All 
& Fourier \newline (\textbf{subexp}) & \cite{Kup05} \\

& All 
& Fourier \newline (\textbf{subexp}) & \cite{Kup13} \\

& Small cases & SV~algorithm & \cite{LBF14} \\

& $N = (q^{d + 1} - 1) / (q - 1)$ 
\newline for specific $q, d$ 
& Hidden shift & \cite{Roe16} \\

& $N=2^n$ 
& Heuristic & \cite{Fuj20} \\

\hline

Other semidirect\newline products 
& $(\mathbb{Z}_2^k \times \mathbb{Z}_2^k) \rtimes_\varphi \mathbb{Z}_2$ 
\newline (i.e., wreath product) 
& Fourier & \cite{RB98} \\

& $\mathbb{Z}_{p}^k \rtimes_\varphi \mathbb{Z}_2$ 
\newline with $p$ prime 
& Fourier & \cite{FIMSS14} \newline (arXiv 2002) \\

& $\mathbb{Z}_{p^r} \rtimes_\varphi \mathbb{Z}_p$
\newline with $p$ odd prime and $r \neq 2$ 
& Black-box & \cite{ILG07} \newline (arXiv 2004) \\

& $\mathbb{Z}_{p^r}^k \rtimes_\varphi \mathbb{Z}_p$
\newline with $p$ prime 
\newline and unique encoding 
& Black-box & \cite{ILG07} \newline (arXiv 2004) \\

& $\mathbb{Z}_N \rtimes_\varphi \mathbb{Z}_p$
\newline with $p$ prime 
\newline and $N / p = \text{poly}(\log N)$ ,
& Fourier & \cite{BCvD05} \\

& $\mathbb{Z}_p^k \rtimes_\varphi \mathbb{Z}_p$
\newline with $p$ prime 
\newline (e.g., Heisenberg groups)
& Fourier & \cite{BCvD05} \\

& $\mathbb{Z}_p \rtimes \mathbb{Z}_q$
\newline with $p, q$ primes,
\newline $q \,|\, (p - 1)$, $p / q = \text{poly}(\log p)$ 
\newline (e.g., $\text{Aff}(1, p)\cong \mathbb{Z}_p \rtimes \mathbb{Z}_{p}^\times$) 
& Fourier & \cite{MRRS07} \newline (arXiv 2005)\\

& $\mathbb{Z}_{2 p^r} \rtimes_\varphi \mathbb{Z}_p$ 
\newline with $p$ odd prime 
& Black-box & \cite{CKL06} \\

& $\mathbb{Z}_{\Pi p_i^{r_i}} \rtimes_\varphi \mathbb{Z}_q$ 
\newline with $p_i, q$ primes 
\newline and $q \,\nmid\, p_i - 1$, $\forall i$
& Black-box & \cite{CKL06} \\

& $\mathbb{Z}_{p^r} \rtimes_\varphi \mathbb{Z}_{p^2}$ 
\newline with $p$ odd prime and $r > 4$ 
& Black-box & \cite{CP07} \\

& $\mathbb{Z}_{\Pi p_i^{r_i}} \rtimes_\varphi \mathbb{Z}_{q^2}$ 
\newline with $p_i, q$ primes 
\newline and $q \,\nmid\, p_i - 1$, $\forall i$
& Black-box & \cite{CP07} \\

& $\mathbb{Z}_p^{k + 1} \rtimes_\varphi \mathbb{Z}_p^k$
\newline with $p$ prime 
\newline (i.e., Weyl-Heisenberg groups) 
& Fourier & \cite{KR08} \\

& $\mathbb{Z}_{p} \rtimes_\varphi \mathbb{Z}_{q^s}$ 
\newline with $p, q$ odd primes, 
\newline $q \,|\, p - 1$, $p / q = \text{poly}(\log p)$
& Fourier & \cite{GPC09} \\

& $\mathbb{Z}_{p^r}\!\rtimes_\varphi \mathbb{Z}_{q^s}$ 
\newline with $p, q$ odd primes,
\newline $p^r / q = \text{poly}(\log p^r)$
& Fourier & \cite{GP11} \\

& $\mathbb{Z}_{p^r} \rtimes_\varphi \mathbb{Z}_{q^s}$ 
\newline with $p \neq q$ primes,
\newline $p^r / q^t = \text{poly}(\log p^r)$,
\newline for specific $0 \leq t \leq s$
& Fourier & \cite{vDD14} \\

& $\mathbb{Z}_{\Pi p_i^{r_i}} \rtimes_\varphi \mathbb{Z}_{q^s}$
\newline with $p_i, q$ primes, $q$ odd,
\newline $q^t \mid p_i - 1$ for one $i$
\newline and specific $0 \leq t \leq s$,
\newline $q \,\nmid\, p_j - 1$ if $j \neq i$
& Black-box & \cite{GFC17} \\

\hline

Other groups 
& $H$ normal in $G$
\newline (e.g., $G$ Dedekind) 
& Fourier & \cite{HRT00} \\

& $\frac{|G|}{|M_G|}\!=\!\exp\big(O\big((\log \log |G|)^{\frac{1}{2}}\big)\big)$
\newline (e.g., $\langle 2^n \rangle \rtimes \langle 3 \rangle$ in $\mathbb Z_{3\cdot 2^n}$)
\newline with $M_G$ intersection of normalizers of $H$ in $G$ 
& Fourier & \cite{GSV01} \newline\\

& Solvable groups with $H$ normal
& Black-box & \cite{IMS01} \\

& Permutation groups with $H$ normal
& Black-box & \cite{IMS01} \\

& Groups with unique encoding
\newline and small commutator
& Black-box & \cite{IMS01} \\

& Groups with unique encoding
\newline and normal elementary abelian 2-subgroup of small index
& Black-box & \cite{IMS01} \\

& Solvable groups with unique encoding and a smoothly solvable commutator subgroup 
& Black-box & \cite{FIMSS14} \newline (arXiv 2002) \\

& $\mathbb{Z}_2^k \times A \times Q_8$
\newline with $A$ abelian,
\newline $Q_8$ quaternion
\newline (i.e., hamiltonian groups)
& Fourier & \cite{HRT03} \\

& Poly-near-hamiltonian groups
& Fourier & \cite{Gav04} \\

& Group extensions by groups with efficient HSP algorithm
& Fourier & \cite{MRRS07} \newline (arXiv 2005)\\

& Extra-special $p$-groups 
& Fourier & \cite{ISS07} \\

& Nil-2 groups
& Fourier & \cite{ISS10} \newline (arXiv 2007) \\

& $G = \text{SL}(2, p^r)$
\newline when $H = G_s:=\{g:gs=s\}$
& Fourier & \cite{DMR10} \\

& $G = \text{PSL}(2, p^r)$
\newline when $H = G_s$
& Fourier & \cite{DMR10} \\

& $G = \text{PGL}(2, p^r)$
\newline when $H = G_s$
& Fourier & \cite{DMR10} \\

& $\text{GL}(n, q)$
\newline with not too small $q < n$
\newline and $H$ Borel
& Fourier & \cite{Iva12} \\

& $\text{GL}(n, q)$
\newline with $H$ parabolic
& Fourier & \cite{DIK14} \\
\hline 
\end{longtable}

\paragraph{Reductions to HSP.}
Here we recall the main problems that can be reduced to an instance of the HSP, highlighting the group $G$ and the function $f$ that hides the subgroup $H$ for each case.

\begin{itemize}

    \item \textbf{Simon's Problem}
    \begin{itemize}
        \item Group: $(\mathbb{Z}_2^n, \oplus)$.
        \item Hiding function: 
        $f\colon \mathbb{Z}_2^n \to \mathbb{Z}_2^m$ such that there is
        $x \in \mathbb{Z}_2^n$ that for each $g \neq g'$ we have $\quad f(g) = f(g')\quad$ if and only if $\quad g' = g \oplus x$.
        \item Hidden subgroup: $\{0, x\}$.
    \end{itemize}

    \item \textbf{Discrete Logarithm Problem (DLP)}
    \begin{itemize}
        \item Group: $(\mathbb{Z}_N \times \mathbb{Z}_N, +)$.
        \item Hiding function: 
        $(g_1,g_2) \mapsto a^{g_1} b^{-g_2}$.
        \item Hidden subgroup: $\langle (x,1) \rangle$.
    \end{itemize}

    \item \textbf{Order Finding Problem (OFP)}
    \begin{itemize}
        \item Group: $(\mathbb{Z}_{\lvert A \rvert}, +)$.
        \item Hiding function: $g \mapsto a^g$.
        \item Hidden subgroup: $\langle x \rangle$.
    \end{itemize}

    \item \textbf{Shortest Vector Problem (SVP) via Dihedral HSP}
    \begin{itemize}
        \item Group: $D_N \cong \mathbb{Z}_N \rtimes_\varphi \mathbb{Z}_2$,
        where $\varphi(x)(n) = (-1)^x n$.
        \item Hiding function: black box.
        \item Hidden subgroup: $H_{N,r} = \langle (r, 1) \rangle$.
    \end{itemize}

    \item \textbf{Permutation code Equivalence Problem (PEP)}
    \begin{itemize}
        \item Group: 
        $\left(\mathrm{GL}_k(\mathbb{F}_q) \times \mathcal{S}_n\right)^2 \rtimes \mathbb{Z}_2$.
        \item Hiding function:
        \[
        ((S_1,P_1),(S_2,P_2), b) \mapsto 
        \begin{cases}
        (S_1^{-1}MP_1,\, S_1^{-1}M'P_2), & b = 0, \\
        (S_1^{-1}M'P_2,\, S_1^{-1}MP_1), & b = 1.
        \end{cases}
        \]
        \item Hidden subgroup:
        \[
        \bigl((H_0, s^{-1}H_0s),0\bigr)
        \cup
        \bigl((H_0s, s^{-1}H_0),1\bigr),
        \]
        where $H_0 = \{(S,P)\mid S^{-1}MP = M\}$.
    \end{itemize}

\end{itemize}

\end{document}